\documentclass{article}
\usepackage{kr}

\usepackage{times}
\usepackage{soul}
\usepackage{url}
\usepackage[hidelinks]{hyperref}
\usepackage[utf8]{inputenc}
\usepackage[small]{caption}
\usepackage{graphicx}
\usepackage{amsmath}
\usepackage{amsthm}
\usepackage{amssymb}
\usepackage{booktabs}
\usepackage{algorithmic}
\usepackage{thm-restate} 

\usepackage[ruled, noend,linesnumbered]{algorithm2e}

\SetAlFnt{\small}
\SetAlCapFnt{\small}
\SetAlCapNameFnt{\small}
\SetAlCapHSkip{0pt}
\IncMargin{-\parindent}
\SetKwInput{KwInput}{Input}
\SetKwInput{KwOutput}{Output}

\usepackage{enumitem}
	
\usepackage{graphicx}
\usepackage{caption}
\usepackage{subcaption}
\usepackage{pgfplotstable}
\usepackage{color}
\usepackage{colortbl}

\usepackage{tikz}
\tikzset{>=latex}

\newcommand{\WKB}{\K_\omega = (\langle \T,\A \rangle, \omega)}
\newcommand{\wkb}{\K_\omega}
\newcommand{\KB}{\K = \langle \T,\A \rangle}
\newcommand{\A}{\mathcal{A}}
\newcommand{\I}{\mathcal{I}}
\newcommand{\J}{\mathcal{J}}
\newcommand{\T}{\mathcal{T}}
\newcommand{\K}{\mathcal{K}}
\newcommand{\Amc}{\A}
\newcommand{\Tmc}{\T}
\newcommand{\Kmc}{\K}
\newcommand{\Imc}{\I}

\newcommand{\Tinf}{\ensuremath{\T_\infty}}
\newcommand{\Ainf}{\ensuremath{\A_\infty}}
\newcommand{\Kinf}{\ensuremath{\K_\infty}}

\newcommand{\NC}{\ensuremath{\mathsf{N_{C}}\xspace}}
\newcommand{\NI}{\ensuremath{\mathsf{N_{I}}\xspace}}
\newcommand{\NR}{\ensuremath{\mathsf{N_{R}}\xspace}}
\newcommand{\cost}[2]{\mi{cost}_{#1}(#2)}
\newcommand{\optc}[1]{\mi{optc}(#1)}
\newcommand{\vio}[2]{vio_{#1}(#2)}
\newcommand{\sat}[2]{\models_{#1}^{#2}}
\newcommand{\dec}[3]{#1QA$^{#2}_{#3}$\xspace}
\newcommand{\BCS}{BCS\xspace}

\newcommand{\ALCO}{\ensuremath{\mathcal{ALCO}}}
\newcommand{\ALCOQu}{\ensuremath{\mathcal{ALCOQ}u}}
\newcommand{\ELbot}{\ensuremath{\mathcal{EL}_\bot}}
\newcommand{\ZOQ}{\ensuremath{\mathcal{ZOQ}}}
\newcommand{\ZOIQ}{\ensuremath{\mathcal{ZOIQ}}}

\def\np{\textsc{NP}\xspace}
\def\conp{co\textsc{NP}\xspace}

\def\deltaptwo{\ensuremath{\Delta^{p}_{2}}\xspace}
\def\deltaptwolog{\ensuremath{\Theta^{p}_{2}}\xspace}
\def\exptime{\textsc{ExpTime}\xspace}
\def\twoexptime{\textsc{2ExpTime}\xspace}

\newcommand{\bsem}[1]{$k$-cost-bounded #1 semantics}
\newcommand{\optsem}[1]{opt-cost #1 semantics}

\newcommand{\mn}[1]{\ensuremath{\mathsf{#1}}}
\newcommand{\mi}[1]{\ensuremath{\mathit{#1}}}

\newcommand{\eg}{e.g.,~}

\newcommand{\ie}{i.e.,~}
\newcommand{\wrt}{w.r.t.~}
\newcommand{\cf}{cf.~}
\newcommand{\wlg}{w.l.o.g.~}
\newcommand{\resp}{resp.~}

\newtheorem{remark}{Remark}
\newtheorem{example}{Example}
\newtheorem{theorem}{Theorem}
\newtheorem{definition}{Definition}
\newtheorem{lemma}{Lemma}
\newtheorem{proposition}{Proposition}

\newcommand{\ind}[1]{#1}

\newcommand{\roleassertion}[3]{\rolestyle{#1}(\ind{#2}, \ind{#3})}

\newcommand{\Knew}{\K^\dagger}
\newcommand{\Anew}{\A^\dagger}
\newcommand{\Tnew}{\T^\dagger}

\newcommand{\signature}{\mathsf{sig}}	

\newcommand{\signatureof}[1][\tbox]{\signature(#1)}

\newcommand{\axiom}[2]{\rolestyle{#1} \sqsubseteq \rolestyle{#2}}

\newcommand{\axand}{\ensuremath{\axiom{A_1 \sqcap A_2}{A}}}

\newcommand{\axexistsleft}{\ensuremath{\axiom{\existsrole{R.B}}{A}}}

\newcommand{\axexistsright}{\ensuremath{\axiom{A}{\existsrole{R.B}}}}

\newcommand{\axnotleft}{\ensuremath{\axiom{\lnot B}{A}}}
\newcommand{\axnotright}{\ensuremath{\axiom{A}{\lnot B}}}

\newcommand{\unfolding}[1][\kb]{\circ}
\newcommand{\unfoldingof}[2][\kb]{\unfolding[#1]}

\newcommand{\deltastar}{\Delta^*}

\newcommand{\domain}[1]{\Delta^{#1}} 

\newcommand{\rolestyle}[1]{\mathrm{#1}}

\newcommand{\existsrole}[1]{\exists \rolestyle{#1}}

\newcommand{\successor}[1][\kb]{\mathsf{succ}^{#1}}

\newcommand{\successorof}[3][\kb]{\successor[#1]_\rolestyle{#3}(#2)}

\newcommand{\individuals}{\mathsf{Ind}} 
\newcommand{\inds}{\individuals}
\newcommand{\indsof}[1]{\inds(#1)}

\newcommand{\exexof}[1]{\exex}
\newcommand{\exex}{\Delta^{\circ}}
\newcommand{\interlace}{f^*}
\newcommand{\interlaceof}[1]{\interlace(#1)}
\newcommand{\interlacing}{\I'}
\newcommand{\interlacingof}[1]{{#1}'}
\newcommand{\interlacingofstar}[1]{{#1}^*}
\newcommand{\interleavingof}[1]{\interlacing}
\newcommand{\exextomod}{f}
\newcommand{\exextomodof}[1]{\exextomod(#1)}

\newcommand{\intertomod}{\sigma}

\title{Cost-Based Semantics for Querying Inconsistent Weighted Knowledge Bases}

\author{%
Meghyn Bienvenu$^{1,2}$\and
Camille Bourgaux$^3$\and
Robin Jean$^1$ \\
\affiliations
$^1$Universit\'{e} de Bordeaux, CNRS, Bordeaux INP, LaBRI, UMR 5800, Talence, France\\
$^2$Japanese-French Laboratory for Informatics, CNRS, NII, IRL 2537, Tokyo, Japan\\ 
$^3$DI ENS, ENS, CNRS, PSL University \& Inria, Paris, France\\
\emails
\{meghyn.bienvenu,robin.jean\}@u-bordeaux.fr,
camille.bourgaux@ens.fr
}

\begin{document}

\maketitle

\begin{abstract}
In this paper, we explore a quantitative approach to querying inconsistent description logic knowledge bases. We consider weighted knowledge bases in which both axioms and assertions have (possibly infinite) weights, which are used to assign a cost to each interpretation based upon the axioms and assertions it violates. Two notions of certain and possible answer are defined by either considering interpretations whose cost does not exceed a given bound or restricting attention to optimal-cost interpretations. Our main contribution is a comprehensive analysis of the combined and data complexity of bounded cost satisfiability and certain and possible answer recognition, for description logics between $\mathcal{EL}_\bot$ and $\mathcal{ALCO}$. 
\end{abstract}


\section{Introduction}
Ontology-mediated query answering (OMQA) is a 
framework for improving data access through the use of an ontology,
which has been extensively studied by the KR and database communities 
\cite{DBLP:journals/jods/PoggiLCGLR08,DBLP:conf/rweb/BienvenuO15,DBLP:conf/ijcai/XiaoCKLPRZ18}. 
Much of the work on OMQA considers ontologies formulated in 
description logics (DLs) \cite{Baader_Horrocks_Lutz_Sattler_2017}. In the DL setting, OMQA consists in finding 
the answers that are logically entailed from the knowledge base (KB), 
consisting of the ABox (data) and TBox (ontology). Due to the use of classical 
first-order semantics, whereby everything is entailed from a contradiction, 
classical OMQA semantics fails to provide informative answers
when the KB is inconsistent. 

The issue of handling inconsistencies, or more generally unwanted consequences,
in DL KBs has been explored from many angles. One solution is to modify the KB 
in order to render it consistent, and there has been significant research on how to 
aid users in the debugging process, e.g.\ by generating 
justifications that pinpoint the sources of the inconsistency
\cite{DBLP:conf/www/ParsiaSK05,DBLP:series/ssw/Penaloza20}.  
This line of work mostly focuses on helping knowledge engineers to debug the 
TBox before deployment in applications, but 
some recent work specifically target ABoxes \cite{DBLP:conf/kr/BaaderK22}. 
However, in an OMQA setting, where
the ABox can be very large and subject to frequent updates,
it is unrealistic to assume that we can always restore consistency.
This has motivated a substantial line of research on 
inconsistency-tolerant semantics to obtain
meaningful answers from inconsistent KBs, surveyed in \cite{biebou,DBLP:journals/ki/Bienvenu20}. 
Many of these semantics are based upon repairs, defined 
as inclusion-maximal subsets of the ABox that are consistent w.r.t.\ the TBox. 
Two of the most commonly considered repair semantics are the AR semantics \cite{Lele}, 
which asks for those answers that hold in every repair, and the brave semantics \cite{DBLP:conf/ijcai/BienvenuR13},
which considers those answers that hold in at least one repair. 
Note that the work on repair-based semantics typically assumes 
that the TBox is reliable, 
which is why repairs are subsets of the ABox, with the TBox left untouched. 
A notable exception is the work of \citeauthor{eit} \shortcite{eit}, which considers generalized notions of repair for existential rule ontologies composed of hard and soft rules,
in which contradictions may be resolved by removing or minimally violating soft rules.

In this paper, we explore a novel quantitative approach to querying 
inconsistent description logic KBs, which combines the idea of soft ontology axioms
from \cite{eit} with a recent cost-based approach to repairing 
databases w.r.t. soft constraints \cite{DBLP:conf/icdt/CarmeliGKLT21}. 
The idea is to associate with every TBox 
axiom and ABox assertion a (possibly infinite) weight. 
`Hard' axioms and assertions, which 
must be satisfied, are assigned a weight of $\infty$, and 
the remaining `soft' axioms and assertions
are assigned weights based upon their reliability, with 
higher weights indicating greater trust. 
The cost of an interpretation is defined by 
taking into account the number of violations of an axiom (assertions can be violated at most once)
and the weights of the violated axioms and assertions. 
When determining the query answers, we shall use
the cost to select a set of interpretations, 
either by considering all interpretations whose cost is below a given threshold, or considering only 
those interpretations having an optimal (i.e. minimum) cost.  
We shall then consider both the certain answers, 
which hold in all of the selected interpretations, 
and the possible answers, which hold w.r.t.\ at least one selected interpretation. 
When restricted to consistent KBs, the optimal-cost certain and possible semantics coincide 
with the classical certain and possible answer semantics, cf.\ \cite{DBLP:conf/aaai/AndolfiCCL24}. 
By varying the cost bounds, we can identify answers that are robust (i.e.\ hold not 
only in all optimal-cost interpretations but also `close-to-optimal' ones) or rank 
candidate answers based upon their incompatibility with the KB.

We perform a comprehensive analysis of the 
complexity of the main 
decision problems in our setting, namely, bounded-cost satisfiability of weighted KBs and 
recognition of certain and possible answers w.r.t.\ the 
set of $k$-cost-bounded or optimal-cost interpretations.  
Our study covers lightweight and expressive description logics,
ranging from $\mathcal{EL}_\bot$ to $\mathcal{ALCO}$, and 
queries given either as instance (IQs) or conjunctive queries (CQs). 
We consider both the combined and data complexity measures, 
as well as the impact of unary and binary encodings of the cost bound 
 and weights. 
Our results are summarized in Table \ref{tab:combinedResults}. 
For combined complexity, most problems are \exptime-complete, 
except for those involving certain answers to CQs, which are \twoexptime-complete. 
For data complexity, we identify problems which are (co)\np-complete,  ${\Theta^{p}_{2}}$-complete, 
and ${\deltaptwo}$-complete\footnote{
$\deltaptwo$ is the class of decision problems solvable in polynomial time with access to an NP oracle, and $\Theta^{p}_{2}$ (aka $\Delta^p_2[log\ n]$) the subclass allowing only logarithmically many NP oracle calls.}, depending on the encoding and maximal value of the weights.  

The paper is organized as follows. Following the preliminaries in Section 2, 
we introduce in Section 3
our formal framework and the associated decision problems. 
Sections 4 and 5 present respectively our combined and data complexity results.
We discuss related work in Section 6 and conclude in Section 7 with some directions for future work. 
Omitted proofs are provided in the appendix.


\section{Preliminaries}

We briefly recall the syntax and semantics of DL. 

\subsubsection*{Syntax} A DL \emph{knowledge base} (KB) $\KB$ consists of an ABox $\A$ and a TBox $\T$, both of which are constructed from three mutually disjoint countable sets $\NC$ of \emph{concept names} (unary predicates), $\NR$ of \emph{role names} (binary predicates), and $\NI$ of \emph{individual names} (constants). The \emph{ABox} 
is a finite set 
of \emph{concept assertions} of the form $A(a)$ with $A \in \NC, a \in \NI$ and \emph{role assertions} of the form $R(a, b)$ with $R \in \NR, a, b \in \NI$. The \emph{TBox} 
is a finite set of axioms 
whose form depends on the DL in question. 
In $\ALCOQu$ TBox axioms are \emph{concept inclusions} $C \sqsubseteq D$ where $C$ and $D$ are complex concepts formed using the following syntax:
\begin{align*}
C := &A \mid \{ a \} \mid \top \mid \bot \mid C \sqcap C \mid C \sqcup C \mid \neg C 
\\ &\mid \exists R.C \mid \forall R.C \mid \leq n R.C \mid \geq n R.C
\end{align*}
where $A \in \NC$, $a \in \NI$, $R \in \NR \cup \{U\}$, with $U$ 
the special \emph{universal role}.\footnote{Usually the universal role cannot occur in qualified number restrictions ($\leq n R.C$ or $\geq n R.C$) but nominals allow us to simulate such number restrictions as explained in \cite{orsi}.} 

The DL $\mathcal{ALCO}$ is the restriction of $\ALCOQu$ disallowing the use of qualified number restrictions ($\leq n R.C$ or $\geq n R.C$) and of the universal role $U$. 
The DL $\mathcal{EL}_\bot$ further disallows 
the use of universal restrictions ($\forall R.C$), negations ($\neg C$), unions ($C \sqcup C$) and nominals ($\{ a \}$). 

We denote by $\inds(\A)$ (\resp $\inds(\K)$) the set of individuals that occur in $\Amc$ (\resp in $\Kmc$), and by $\signatureof$ (\resp $\signatureof[\K]$) the set of concept and role names that occur in  $\Tmc$ (\resp in $\K$).

\subsubsection*{Semantics} An \emph{interpretation} has the form $\I=(\Delta^{\I},{.}^{\I})$, where the \emph{domain} $\Delta^{\I}$ is a non-empty set and ${.}^{\I}$ maps each $a \in \NI$ to $a^{\I} \in \Delta^{\I}$, each $A \in \NC$ to $A^{\I} \subseteq \Delta^{\I}$, each $R \in \NR$ to $R^{\I} \subseteq \Delta^{\I} \times \Delta^{\I}$ and interprets the universal role $U$ by $U^{\I} = \Delta^{\I} \times \Delta^{\I}$.  The function  ${.}^{\I}$ is extended to general concepts, \eg $(\exists R.D)^{\I}=\{c \mid \exists d \in D^\I:(c, d) \in R^{\I}\}$; $\{ a \}^\I = \{a^\I \}$; $\top^\I = \Delta^\I$; $\bot^\I = \emptyset$; $(\leq n R.C)^\I = \{ c \mid \#\{ d \in C^\I \mid (c,d) \in R^\I \} \leq n \}$ and $(\geq n R.C)^\I = \{ c \mid \#\{ d \in C^\I \mid (c,d) \in R^\I \} \geq n \}$. 
An interpretation $\I$ satisfies an assertion $A(a)$ (\resp $R(a, b)$) if $a \in A^{\I}$ (\resp $(a, b) \in R^{\I}$);  we thus make a weak version of the \emph{standard names assumption} (SNA).\footnote{
The usual SNA requires that $a^\I=a$ for every $a\in\NI$, hence that $\NI \subseteq \Delta^\I$, so all interpretations have an infinite domain. To be able to bound the size of interpretations, we adopt this `weak' version of the SNA, used \eg by \citeauthor{lutz2023querying} \shortcite{lutz2023querying}.   
} 
$\I$ satisfies an inclusion $C \sqsubseteq D$ if $C^{\I} \subseteq D^{\I}$ and $\{a\}^\I=\{a\}$ for every nominal occurring in $C$ or $D$. 
We write $\I \models \tau$ (resp.\ $\I \models \alpha$) to indicate that $\I$ satisfies an axiom $\tau$ (resp.\ assertion $\alpha$).
An interpretation $\I$ is a \emph{model} of $\KB$, denoted $\I\models\K$, if $\I$ satisfies all inclusions in $\T$ ($\I\models\T$) and all assertions in $\A$ ($\I\models\A$). A KB $\K$ is \emph{consistent} if it has a model.

\subsubsection*{Queries} 
We consider \emph{conjunctive queries} (CQs) which take the form $\exists \vec{y} \psi$, where $\psi$ is a conjunction of atoms of the forms $A(t)$ or $R(t, t')$, where $t, t'$ are variables or individuals, and $\vec{y}$ is a tuple of variables from $\psi$. A CQ is called \emph{Boolean} (BCQ) if all of its variables are existentially quantified; a CQ consisting of a single atom is an \emph{instance query} (IQ). When we use the generic term query, we mean a CQ. For a BCQ \emph{q} and an interpretation $\I$, we denote by $\I \models q$ the fact that $\I$ satisfies \emph{q}. A BCQ \emph{q} is \emph{entailed} from $\K$, written $\K \models q$, if $\I \models q$ for every model $\I$ of $\K$. 
A BCQ $q$ is \emph{satisfiable w.r.t.\ $\K$} if there exists a model $\I$ of $\K$ such that $\I \models q$. 
For a non-Boolean CQ $q[\vec{x}]$ with free variables $\vec{x}=(x_1, \ldots, x_k)$, a tuple of individuals $\vec{a}=(a_1, \ldots, a_k)$ is a \emph{certain answer} for $q[\vec{x}]$ \wrt $\K$ just in the case that $\K \models q[\vec{a}]$, where $q[\vec{a}]$ is the BCQ obtained by replacing each $x_i$ by $a_i$. Tuple $\vec{a}$ is said to be a \emph{possible answer} for $q[\vec{x}]$ \wrt $\K$  if the BCQ $q[\vec{a}]$ is satisfiable w.r.t.\ $\K$. Observe that certain and possible answer recognition corresponds to BCQ entailment and satisfiability respectively. 

To simplify the presentation, we shall focus on BCQs. However, 
all definitions and results are straightforwardly extended to non-Boolean queries, 
and we shall thus sometimes speak of `query answers' when providing intuitions.


\section{Weighted Knowledge Bases} \label{weighted knowledge bases}

We consider a quantitative way of integrating the notion of soft constraints by giving weights to axioms and assertions. Intuitively, these weights represent penalties associated to each violation of the axioms or assertions. They will allow us to assign a cost to interpretations based upon the axioms and assertions they violate, and use this cost to select which interpretations to consider when answering queries. 

\begin{definition}
    A \emph{weighted knowledge base} (WKB) $\WKB$ consists of 
a knowledge base  $\langle \T ,\A\rangle$ 
and a cost function $\omega: \T \cup \A \mapsto \mathbb{N}_{> 0} \cup \{ + \infty \}$. 
     We denote by $\T_\infty$ (\resp $\A_\infty$) the set of TBox axioms (\resp ABox assertions) that have an infinite cost and let $\K_\infty=\langle\T_\infty,\A_\infty\rangle$. We sometimes use   $\omega_\chi$ as a shorthand for $\omega(\chi)$.     
\end{definition}

\begin{example}\label{ex:running}
Consider the following WKB about visa requirements to enter some country $c$: $\WKB$ where $\Tmc=\{\tau_1,\tau_2,\tau_3\}$, $\Amc=\{\alpha_1,\alpha_2\}$ and  
 \begin{align*} 
\tau_1=&\mn{Visa}\sqcap\mn{NoVisa}\sqsubseteq \bot && \omega(\tau_1)=\infty\\
\tau_2=&\exists \mn{hasNat}.\{c\}\sqcap \exists \mn{hasNat}.\{b\}\sqsubseteq \bot && \omega(\tau_2)=\infty\\
\tau_3=&\forall \mn{hasNat}.\neg\{c\}\sqsubseteq \mn{Visa}&& \omega(\tau_3)=1\\
\alpha_1=&\mn{hasNat}(p,b)&& \omega(\alpha_1)=1\\  
\alpha_2=&\mn{NoVisa}(p)&& \omega(\alpha_2)=2
 \end{align*}
 Two `absolute' constraints $\tau_1$ and $\tau_2$ express that one cannot both need a visa ($\mn{Visa}$) and not need one ($\mn{NoVisa}$) and that it is not possible to have both nationalities ($\mn{hasNat}$) $c$ and~$b$. A `soft' constraint $\tau_3$ expresses that someone that does not have nationality $c$ normally needs a visa. The ABox states that a person $p$ has nationality $b$ and does not need a visa, and the second assertion is more reliable than the first one.
 \end{example}

To measure how far an interpretation is from being a model of the KB, we rely on the following sets of \emph{violations}. 

\begin{definition}\label{defvio1}
The \emph{set of violations of a concept inclusion $B \sqsubseteq  C$} in an interpretation $\I$ is the set $$\vio{B \sqsubseteq C}{\I} = (B \sqcap \neg C)^{\I}.$$
The \emph{violations of an ABox $\A$} in an interpretation $\I$ are  
$$\vio{\A}{\I} = \{ \alpha \in \A \mid \I \not \models \alpha \}.$$
\end{definition}

These sets of violations can be used to associate a \emph{cost} to interpretations, by taking into account the weights assigned by the WKB to the violated inclusions and assertions.

\begin{definition}
Let $\WKB$ be a WKB. 
The \emph{cost of an interpretation} $\I$ \wrt $\wkb$ is defined by:         
$$\cost{\wkb}{\I}  = \sum_{\tau \in \T} \omega_\tau|\vio{\tau}{\I}| + \sum_{\alpha \in \vio{\A}{\I}}\omega_\alpha.$$
We say that $\wkb$ is \emph{$k$-satisfiable} if there exists an interpretation $\I$ with $\cost{\wkb}{\I}\leq k$ and define  
the \emph{optimal cost} of $\wkb$ as $\optc{\wkb}  = \min_{\I} (\cost{\wkb}{\I})$. 
\end{definition}

\begin{remark}
Note that $\cost{\wkb}{\I}$ will be $\infty$ if any infinite-weight assertion or inclusion is violated in $\I$ and/or if any inclusion has an infinite set of violations in $\I$. 
\end{remark}

\begin{example}[Ex.\ref{ex:running} cont'd]
Consider the following interpretations over domain $\Delta^\I=\{p,b,c\}$ that correspond to different possibilities for $p$'s nationalities and need for a visa.
\begin{itemize}
\item Case $p$ has nationality $b$ and needs a visa: $\mn{hasNat}^{\I_b^v}=\{(p,b)\}$, $\mn{Visa}^{\I_b^v}=\{p\}$ and $\mn{NoVisa}^{\I_b^v}=\emptyset$. 
${\I_b^v}$ violates only $\alpha_2$ so $\cost{\wkb}{{\I_b^v}} = 2$.
\item Case $p$ has nationality $b$ and does not need a visa: $\mn{hasNat}^{\I_b^n}=\{(p,b)\}$, $\mn{Visa}^{\I_b^n}=\emptyset$ and $\mn{NoVisa}^{\I_b^n}=\{p\}$. 
${\I_b^n}$ violates only $\tau_3$ so \mbox{$\cost{\wkb}{{\I_b^n}} = 1$.}
\item Case $p$ has nationality $c$ and needs a visa: $\mn{hasNat}^{\I_c^v}=\{(p,c)\}$, $\mn{Visa}^{\I_c^v}=\{p\}$ and $\mn{NoVisa}^{\I_c^v}=\emptyset$. 
${\I_c^v}$ violates only $\alpha_1$ and $\alpha_2$ so $\cost{\wkb}{{\I_c^v}} = 3$.
\item Case $p$ has nationality $c$ and does not need a visa: $\mn{hasNat}^{\I_c^n}=\{(p,c)\}$, $\mn{Visa}^{\I_c^n}=\emptyset$ and $\mn{NoVisa}^{\I_c^n}=\{p\}$. 
${\I_c^n}$ violates only $\alpha_1$ so $\cost{\wkb}{{\I_c^n}} = 1$.
\item Case $p$ has nationality $b$ and $c$ and does not need a visa: $\mn{hasNat}^{\I_{bc}^n}=\{(p,b), (p,c)\}$, $\mn{Visa}^{\I_{bc}^n}=\emptyset$ and $\mn{NoVisa}^{\I_{bc}^n}=\{p\}$. 
${\I_{bc}^n}$ violates $\tau_2$ so $\cost{\wkb}{{\I_{bc}^n}} = \infty$.
\end{itemize}
Since $\Kmc_\omega$ is inconsistent and the smallest weight is $1$, it follows that $\I_b^n$ and $\I_c^n$ are of optimal cost and $\optc{\wkb}  = 1$.
\end{example}

It is now possible to define variants of the classical certain and possible answers, by considering either only interpretations whose cost does not exceed a given bound, or only optimal-cost interpretations. For simplicity, we state the definitions in terms of BCQ entailment. 

\begin{definition}
Let $q$ be a BCQ and $\WKB$ be a WKB.
We say that $q$ is entailed by $\wkb$ under 
\begin{itemize}
\item \emph{\bsem{certain}}, written $\wkb \sat{c}{k} q$, if $\I \models q$ for every interpretation $\I$ with $\cost{\wkb}{\I} \leq k$;

\item \emph{\bsem{possible}}, written $\wkb \sat{p}{k} q$, if $\I \models q$ for some interpretation $\I$ with $\cost{\wkb}{\I} \leq k$;

\item \emph{\optsem{certain}}, written $\wkb \sat{c}{opt} q$, if $\I \models q$ for every interpretation $\I$ with $\cost{\wkb}{\I} = \optc{\wkb}$;	
	
\item \emph{\optsem{possible}}, written $\wkb \sat{p}{opt} q$, if $\I \models q$ for some interpretation $\I$ with $\cost{\wkb}{\I} = \optc{\wkb}$.
\end{itemize}	
\end{definition}

\begin{example}[Ex.\ref{ex:running} cont'd]
Since $\optc{\wkb}  = 1$, weights of axioms different from $\tau_3$ and $\alpha_1$ are greater than $1$ and $\I^n_b$ and $\I^n_c$ are interpretations of cost $1$ that violate $\tau_3$ and $\alpha_1$ respectively, it follows that interpretations of optimal cost violate exactly one axiom in $\{\tau_3,\alpha_1\}$. In particular, they all satisfy $\alpha_2$, \ie $\wkb \sat{c}{opt} \mn{NoVisa}(p)$. 
Since all interpretations of minimal cost satisfy $\tau_1$, it follows that $\wkb \not\sat{p}{opt} \mn{Visa}(p)$. 
On the other hand, we obtain that $\wkb \sat{p}{opt} \mn{hasNat}(p,b)$ (because of ${\I_{b}^n}$) and $\wkb \sat{p}{opt} \mn{hasNat}(p,c)$ (because of ${\I_{c}^n}$). 
 
 If we now consider interpretations of cost bounded by $2$, we obtain that $\wkb \not\sat{c}{2} \mn{NoVisa}(p)$ and $\wkb \sat{p}{2} \mn{Visa}(p)$ (because of ${\I_{b}^v}$), hence we cannot conclude anymore whether $p$ needs a visa or not using the certain semantics. 
 However, we can still exclude some statements even under possible semantics. For example, $\wkb \not\sat{p}{2} \mn{hasNat}(p,c)\wedge \mn{Visa}(p)$, since this holds only in interpretations of cost at least 3. 
\end{example}

\begin{table*}
\centering
\begin{tabular}{lccccccc}
\toprule
& \BCS  & \dec{I}{b}{p}, \dec{C}{b}{p} & \dec{I}{b}{c} & \dec{C}{b}{c}  & \dec{I}{opt}{p}, \dec{C}{opt}{p} & \dec{I}{opt}{c} &  \dec{C}{opt}{c} \\
\midrule
Combined & \exptime  & \exptime & \exptime & \twoexptime & \exptime & \exptime & \twoexptime    \\ 
\midrule
Data & \np & \np &  \conp & \conp & ${\deltaptwo}^*$ / ${\Theta^{p}_{2}}^\dag$ & ${\deltaptwo}$-hard$^*$ / ${\Theta^{p}_{2}}^\dag$ & ${\deltaptwo}$-hard$^*$ / ${\Theta^{p}_{2}}^\dag$ \\
\bottomrule
\end{tabular}
\caption{Overview of complexity results for description logics between $\mathcal{EL}_\bot$ and $\mathcal{ALCO}$. All bounds are tight except the two `-hard' cases. Lower bounds hold even if the weights (and the input integer in the case of combined complexity) are encoded in unary, except those marked with $^*$. Upper bounds hold even if the weights (and the input integer in the case of combined complexity) are encoded in binary. 
$^\dag$: ${\Theta^{p}_{2}}$-complete 
 if the finite weights on the assertions are either bounded (independently from $|\A|$), or encoded in unary. }
\label{tab:combinedResults}
\end{table*}

When the underlying KB is consistent, the certain and possible optimal-cost semantics 
coincide with classical query entailment and query satisfiability (or classical certain and possible answers in the case of non-Boolean queries): 

\begin{restatable}{proposition}{consistentCase}
Let $\wkb$ be such that $\optc{\wkb}=0$. Then:
\begin{itemize}
\item $\wkb \sat{c}{opt} q$ iff $\K \models q$
\item $\wkb \sat{p}{opt} q$ iff $q$ is satisfiable w.r.t.\ $\K$
\end{itemize}
\end{restatable}

It is also interesting to consider how the $k$-cost-bounded semantics vary with different values of $k$:

\begin{restatable}{proposition}{kVariesProp} \label{k-VariesProp}
Consider a WKB $\wkb$, BCQ $q$, and $k \geq 0$.  
\begin{itemize}
\item If $\wkb \sat{c}{k} q$, then $\wkb \sat{c}{k'} q$ for every $0 \leq k' \leq k$
\item If $\wkb \not \sat{p}{k} q$, then $\wkb \not \sat{p}{k'} q$ for every $0 \leq k' \leq k$
\end{itemize}
Moreover, $\wkb \not \sat{p}{k} q$ and $\wkb \sat{c}{k} q$ if $k < \optc{\wkb}$. 
\end{restatable}

The preceding result shows, unsurprisingly, that $k$-cost-bounded semantics are only informative for $k \geq \optc{\wkb}$.  
Increasing $k$ beyond $\optc{\wkb}$ leads to fewer and fewer queries being entailed under the $k$-cost-bounded certain semantics, which may be useful in identifying query answers that are robust in the sense that they continue to hold even if we consider a larger set of `close-to-optimal' interpretations. By contrast, as $k$ grows, so does the set of entailed queries under  $k$-cost-bounded possible semantics. Being quite permissive, the opt-cost and $k$-cost-bounded possible semantics will entail many queries, and thus are not suitable replacements for standard (certain answer) querying semantics. 
Instead, non-entailment under these semantics can serve to eliminate or rank candidate tuples of individuals (or the corresponding instantiated Boolean queries) based upon how incompatible they are w.r.t. the expressed information.

\subsubsection*{Relationship With Preferred Repair Semantics} We show that \optsem{certain} generalizes the $\leq_\omega$-AR semantics defined by \citeauthor{biebougoa}~\shortcite{biebougoa} for KBs with weighted ABoxes, where $\omega : \A \rightarrow \mathbb{N}_{> 0}$ models the reliability of the assertions while the TBox axioms are considered absolute. 
In this context, $\leq_\omega$-repairs are subsets of the ABox consistent with the TBox and maximal for the preorder defined over ABox subsets by $ \A_1 \leq_\omega \A_2$ if $\sum_{\alpha \in \A_1}\omega_\alpha \leq \sum_{\alpha \in \A_2}\omega_\alpha$. 
A BCQ $q$ is entailed under $\leq_\omega$-AR (resp. $\leq_\omega$-brave) semantics if $\langle \T,\A' \rangle \models q$ for every (\resp some) $\leq_\omega$-repair $\A'$ of $\A$. 

\begin{restatable}{proposition}{RelPreferredAR}
   Let $\WKB$ be a WKB such that $\T$ is satisfiable, $\omega(\tau) = \infty$ for every $\tau \in \T$ and $\omega(\alpha)\neq \infty$ for every $\alpha\in\A$, and let $q$ be a BCQ. 
    $$ \wkb \sat{c}{opt} q \iff  \langle \T, \A \rangle \models_{\leq_\omega\text{-AR}} q$$
\end{restatable}

\begin{proof}[Proof sketch]
Since $\T$ is satisfiable, there is a model $\I$ of $\T$, and since $\omega(\alpha)\neq \infty$ for every $\alpha\in\A$, $\cost{\wkb}{\I}\neq \infty$. It follows that $\optc{\wkb} \neq \infty$ and that every $\I$ of optimal cost is such that $\I\models \T$ and $\cost{\wkb}{\I}= \sum_{\alpha \in vio_\A(\I)} \omega_\alpha$. 

Hence, for every $\I$ of optimal cost, $\A' = \A \setminus vio_\A(\I)$ is a $\leq_\omega$-repair. Indeed, since $\I\models\langle\T,\A'\rangle$, $\langle\T,\A'\rangle$ is consistent. Moreover, $\sum_{\alpha \in \A'} \omega_\alpha = \sum_{\alpha \in \A} \omega_\alpha - \cost{\wkb}{\I}$ and $\cost{\wkb}{\I}$ is minimal, so $\sum_{\alpha \in \A'} \omega_\alpha$  is maximal. 
It also follows that every $\leq_\omega$-repair $\A'$ is such that $\sum_{\alpha \in \A'} \omega_\alpha = \sum_{\alpha \in \A} \omega_\alpha - \optc{\wkb}$. 
\end{proof}

Note however that \optsem{possible}, does \emph{not} generalize $\leq_\omega$-brave, but only over-approximates it: 
$$  \langle \T, \A \rangle \models_{\leq_\omega\text{-brave}} q \implies \wkb \sat{p}{opt} q. $$ 
Indeed, we have shown that to each $\leq_\omega$-repair corresponds at least one interpretation with optimal cost but given an interpretation $\I$ with optimal cost  \wrt $\WKB$, if $B \in \NC\setminus\signatureof[\Kmc]$ and $b \in \NI\setminus\inds(\Kmc)$, then one can add $b^\I$ to $B^\I$ without changing the cost of $\I$ \wrt $\wkb$, so $\wkb \sat{p}{opt} B(b)$ while $\langle \T, \A \rangle \not\models_{\leq_\omega\text{-brave}} B(b)$.

\subsubsection*{Decision Problems} 

In our complexity analysis, we will consider the following decision problems:
\begin{itemize}
\item \emph{Bounded cost satisfiability}  (\BCS) takes as input a WKB $\WKB$ and an integer $k$ and decides whether there exists an interpretation $\I$ with $\cost{\wkb}{\I} \leq k$.
\item \emph{Bounded-cost certain (\resp possible) BCQ entailment} (\dec{C}{b}{c} (\resp \dec{C}{b}{p})) takes as input a WKB $\WKB$, a BCQ $q$ and an integer $k$ and decides whether $\wkb \sat{c}{k} q$ (\resp  $\wkb \sat{p}{k} q$).
\item \emph{Optimal-cost certain (\resp possible) BCQ entailment} (\dec{C}{opt}{c} (\resp \dec{C}{opt}{p})) takes as input a WKB $\WKB$ and a BCQ $q$ and decides whether $\wkb \sat{c}{opt} q$ (\resp  $\wkb \sat{p}{opt} q$). 
\end{itemize}
We will also consider the restrictions of the Boolean query entailment problems to the case of instance queries, denoted by \dec{I}{b}{c}, \dec{I}{b}{p}, \dec{I}{opt}{c} and \dec{I}{opt}{p} respectively.

\subsubsection*{Complexity Measures} 
It is customary to consider \emph{combined complexity} and \emph{data complexity} when studying decision problems related to query answering over DL KBs. Data complexity considers only the size of the ABox while combined complexity takes into account the size of the whole input. In the case of WKBs, we consider the assertion weights as part of the ABox and inclusion weights as part of the TBox. We will use the following notation: given a WKB $\WKB$, $|\A|$ (\resp $|\T|$, $|\Kmc|$) is the length of the string representing $\Amc$ (\resp $\Tmc$, $\Kmc$), where elements of $\NC$, $\NR$ and $\NI$ are considered of length one, and $|\A_\omega|$ (\resp $|\T_\omega|$, $|\Kmc_\omega|$) is the length of the string representing the set $\{(\alpha,\omega(\alpha))\mid \alpha\in\Amc\}$ (\resp $\{(\tau,\omega(\tau))\mid \tau\in\Tmc\}$, $\{(\chi,\omega(\chi))\mid \chi\in\Tmc\cup\Amc\}$), where elements of $\NC$, $\NR$ and $\NI$ are considered of length one and weights are encoded either in unary or in binary. Note that if the TBox contains qualified number restrictions, the numbers can also be encoded in unary or binary. We will also make this encoding distinction for the integer $k$ taken as input by some of the decision problems we consider. If $|k|$ denotes the size of the encoding of $k$, $k=|k|$ when encoded in unary and $k\leq 2^{|k|}$ when encoded in binary. Finally, for a BCQ $q$, $|q|$ is the length of the string representing $q$ where elements of $\NC$, $\NR$, $\NI$ and variables are considered of size one. Note that when we use $|\cdot|$ over a set which is not a (weighted) ABox or TBox, we simply means the set cardinality.


\section{Combined Complexity}\label{sec:combined}
In this section we study the combined complexity of bounded cost satisfiability and certain and possible answer recognition, for DLs between $\mathcal{EL}_\bot$ and $\mathcal{ALCO}$. The first line of Table~\ref{tab:combinedResults} gives an overview of the results.

\subsection{Upper Bounds}

To characterize the cost of interpretations, we define the notion of \emph{$k$-configuration}. Intuitively, a $k$-configuration specifies how to allocate a cost of $k$ between possible violations. 

\begin{definition}[$k$-configuration]
Let $\WKB$ be a WKB and $k$ be an integer.
A \emph{$k$-configuration} for $\wkb$ is a function $\gamma : \T \cup \A \mapsto \mathbb{N}$ such that:
            \begin{itemize}
                \item $\gamma(\tau) \in \mathbb{N}$ for every $\tau\in \T$,
                \item $\gamma(\alpha) \in \{0,1\}$ for every $\alpha \in \A$,
                \item $\sum_{\chi\in \T \cup \A} \gamma(\chi)\omega_\chi \leq k$.           
    		\end{itemize}
An interpretation $\I$ satisfies the $k$-configuration $\gamma$ if $ |\vio{\tau}{\I}| \leq \gamma(\tau)$ for every $\tau\in \T$ and $\I\models\alpha$ for every $\alpha \in \A$ such that $\gamma(\alpha)= 0$. 
\end{definition}

The definition of $k$-configurations implies in particular that $\gamma(\chi) = 0$ for every $\chi \in \T_\infty\cup \A_\infty$.
	
\begin{lemma} \label{lem:kconfig}
Let $\wkb$ be a WKB and $\I$ be an interpretation. 
$$\cost{\wkb}{\I} = \min \{ k \mid \exists\gamma \text{ $k$-configuration s.t. } \I \text{ satisfies } \gamma \}$$
\end{lemma}

\begin{proof}
If $\I$ satisfies a $k$-configuration $\gamma$, $\cost{\wkb}{\I} \leq k$. Indeed, for every $\tau\in \T$, $ |\vio{\tau}{\I}| \leq \gamma(\tau)$, and for every $\alpha\in \vio{\A}{\I}$, $\gamma(\alpha)=1$ because $\gamma(\alpha)=0$ implies $\I\models\alpha$. Thus 
$\cost{\wkb}{\I} \leq \sum_{\tau \in \T} \gamma(\tau)\omega_\tau + \sum_{\alpha \in \A}\gamma(\alpha)\omega_\alpha \leq k$.   
Moreover, if $\cost{\wkb}{\I} = k$, we can define a $k$-configuration $\gamma$ such that $\I$ satisfies $\gamma$ by setting $\gamma(\tau) = |\vio{\tau}{\I}|$ for every $\tau \in \T$, and $\gamma(\alpha) = 0$ if $\I\models\alpha$, $\gamma(\alpha) = 1$ otherwise for every $\alpha \in \A$.
\end{proof}

We now define a new KB in a more expressive DL in such a way that the models of the new KB will be interpretations that satisfy a given $k$-configuration.

Given a concept inclusion $\tau=B \sqsubseteq C$ we define the \emph{violation concept} $V_\tau := B \sqcap \neg C$ such that for every interpretation $\I$, it holds that $\vio{\tau}{\I} = V_\tau^\I$.

\begin{definition} \label{defKB}
    Let $\WKB$ be an $\mathcal{ALCO}$ WKB, $k$ an integer and $\gamma$ a $k$-configuration for $\wkb$.   
    We define the $\ALCOQu$ KB $\K_{\gamma}= \langle \T_{\gamma}, \A_{\gamma} \rangle$ associated to $\wkb$ and $\gamma$ as:
\begin{align*}
\T_{\gamma} = & \T_\infty \cup \{\top \sqsubseteq  \: \leq\gamma(\tau) \: U.V_\tau\mid \tau\in \T\setminus\T_\infty\} \\
\A_{\gamma} = & \{ \alpha \in \A \mid \:\: \gamma(\alpha) = 0 \}.
\end{align*}
\end{definition}

\begin{proposition}\label{prop:satisfactionConfigKB}
    Let $\wkb$ be a WKB and $\gamma$ be a $k$-configuration for $\wkb$. For every interpretation $\I$, 
    $\I \models \K_{\gamma}$ iff $\I$ satisfies $\gamma$. 
\end{proposition}
\begin{proof}
Suppose $\I \models \K_\gamma$. For every $\tau\in\T\setminus\T_\infty$, since $\I \models \top \sqsubseteq \: \leq\gamma(\tau) U.V_\tau$, then $|vio_\tau(\I)| = |V^{\I}_\tau| \leq \gamma(\tau)$. For every $\tau\in\T_\infty$, since $\I\models \T_\infty$,  $|vio_\tau(\I)|=0=\gamma(\tau)$. Finally, as $\I \models \A_\gamma $, $\I$ satisfies all $\alpha \in \A$ such that $\gamma(\alpha)=0$.

Conversely, suppose that $\I$ satisfies $\gamma$. For every $ \tau \in \T\setminus\T_\infty , \: |V^{\I}_\tau| = |\vio{\tau}{\I}| \leq \gamma(\tau)$ thus $\I \models \top \sqsubseteq \: \leq \gamma(\tau) U.V_\tau$. For every $\tau \in \T_\infty $,  $|\vio{\tau}{\I}| \leq \gamma(\tau) = 0$ thus $\I \models \tau$. Therefore $\I \models \T_\gamma$. As $\I$ satisfies all $\alpha \in \A$ such that $\gamma(\alpha)= 0$, we also have $\I \models \A_\gamma$, so $\I\models\K_\gamma$.
\end{proof}

This construction allows us to decide bounded cost satisfiability via $\ALCOQu$ satisfiability.

\begin{theorem}\label{thm:upperboundBCS}
    \BCS for $\mathcal{ALCO}$ is in \exptime in combined complexity (even if the bound $k$ and the weights are encoded in binary). 
\end{theorem}
\begin{proof}
Let $\WKB$ be a WKB and $k$ be an integer. By Lemma \ref{lem:kconfig} and Proposition \ref{prop:satisfactionConfigKB}, $\wkb$ is $k$-satisfiable iff there exists a $k$-configuration $\gamma$ such that $\K_\gamma$ is satisfiable. The number of $k$-configurations $\gamma$ is bounded by $(k+1)^{|\T|}2^{|\A|}$ (hence by $2^{|k+1||\T|+|\A|}$ if $k$ is encoded in binary and $|k+1|$ is the length of the encoding of $k+1$) as there are at most $k+1$ possibilities for the value of $\gamma(\tau)$ for $\tau \in \T$ and 2 possibilities for the value of $\gamma(\alpha)$ for $\alpha \in \A$. Moreover, for a given $\gamma$, $\K_\gamma$ is of polynomial size and can be constructed in polynomial time \wrt $|\wkb|$ and $|k|$ by encoding the number restrictions in binary (since the numbers occurring in such restrictions are bounded by $k$). Therefore, as satisfiability in $\mathcal{ZOQ}$ (which extends $\ALCOQu$) is in \exptime even with binary encoding in number restrictions \cite{ceo}, checking for every $k$-configuration $\gamma$ whether $\K_\gamma$ is satisfiable is a decision procedure for \BCS that runs in exponential time \wrt combined complexity. 
Note that the complexity results for $\ALCOQu$ apply even if they are shown without the standard name assumption 
because $\Kmc_\gamma$ is satisfiable under our weak SNA iff $\Kmc_\gamma\cup\{ \{a\}\sqcap\{b\}\!\sqsubseteq\! \bot \mid \!a, b\! \in\! \inds(\K),a\! \neq\! b\}$ is satisfiable without any assumption on the  interpretation of individuals. 
\end{proof}

To prove the upper bounds on query entailment, we need to first show some results on the computation of the optimal cost of an $\mathcal{ALCO}$ WKB. 
Since the number of violations of a concept inclusion in an interpretation $\I$ is bounded by the cardinality of its domain $\Delta^\I$, 
the following proposition is useful to bound the optimal cost of a WKB.

\begin{restatable}{proposition}{PropBoundedModelALCO}\label{prop:boundedModelALCO}
	Let $\K$ be an $\mathcal{ALCO}$ KB. If $\K$ is satisfiable, then it has a model $\I$  such that $|\Delta^\I| \leq |\inds(\K)| + 2^{|\T|}$. 
\end{restatable}
\begin{proof}[Proof sketch]
We adapt the proof of $\mathcal{ALC}$ bounded model property by \citeauthor{Baader_Horrocks_Lutz_Sattler_2017}~\shortcite{Baader_Horrocks_Lutz_Sattler_2017}. It is based on the notion of filtration that `merges' elements that belong to the same concepts and is easily extended to handle nominals. 
\end{proof}

The following lemma is a consequence of Proposition~\ref{prop:boundedModelALCO} and the definition of the cost of an interpretation. 

\begin{lemma}\label{lem:boundOptc}
The optimal cost for an $\mathcal{ALCO}$ WKB $\wkb$ (such that $\K_\infty$ is satisfiable) is exponentially bounded in $|\wkb|$ (even if  the weights are encoded in binary).
\end{lemma}
\begin{proof}
Let $\WKB$ be an $\ALCO$ WKB such that $\Kinf$ is satisfiable. 
By Proposition \ref{prop:boundedModelALCO}, there exists a model $\I$ of $\Kinf$ such that $|\Delta^\I| \leq |\inds(\K)| + 2^{|\T|} := l$.
\begin{align*}
		\cost{\wkb}{\I} &= \sum_{\tau \in \T} \omega_\tau|vio_\tau(\I)| + \sum_{\alpha \in vio_\A(\I)}\omega_\alpha
		\\ &\leq l(\sum_{\tau \in \T \setminus \Tinf} \omega_\tau) + \sum_{\A \setminus \Ainf}\omega_\alpha        
		\\ &\leq l|\T|\max_{\tau\in\T\setminus\T_\infty}(\omega_\tau) + |\A|\max_{\alpha\in\A\setminus\A_\infty}(\omega_\alpha) 
\end{align*}
It follows that $\optc{\wkb}\leq \cost{\wkb}{\I} \leq L$ where $L:=(|\T||\inds(\K)|+|\T| 2^{|\T|})\max_{\tau\in\T\setminus\T_\infty}(\omega_\tau) + |\A|\max_{\alpha\in\A\setminus\A_\infty}(\omega_\alpha)$. 
Moreover, since 
the maximal (finite) weights are at most exponential in $|\wkb|$ (even if weights are encoded in binary), $L$ is exponential in $|\wkb|$. 
\end{proof}

Since the optimal cost is exponentially bounded and \BCS is in \exptime, we obtain the following result.

\begin{lemma}\label{lem:compOpt}
Computing the optimal cost of an $\mathcal{ALCO}$ WKB $\wkb$ can be done in exponential time in the size of the WKB $|\wkb|$ (even if  the weights are encoded in binary).
\end{lemma}

\begin{proof}
Let $\WKB$ be an $\ALCO$ WKB. 
If $\Kmc_\infty$ is not satisfiable, $\optc{\wkb}=\infty$. Otherwise, by Lemma~\ref{lem:boundOptc}, $\optc{\wkb}\leq L$ for some $L:=2^{p(|\wkb|)}$ where $p$ is a polynomial function. 
To compute $\optc{\wkb}$, one can check for every $0\leq i\leq L$ whether there exists 
$\I$ with $\cost{\wkb}{\I} \leq i$. 
By Theorem~\ref{thm:upperboundBCS}, each call to \BCS takes 
exponential time \wrt $|\wkb|$ and the size of the binary encoding of $i$, which is bounded by $p(|\wkb|)$. 
The whole computation 
thus takes exponential time \wrt $|\wkb|$. 
\end{proof}

We show that BCQ entailment (hence also BIQ entailment) under our variants of the possible semantics can be decided through an exponential number of calls to \BCS. 

\begin{theorem}\label{th:combined-upper-cqa-p}
	\dec{C}{b}{p} and \dec{C}{opt}{p} for $\mathcal{ALCO}$ are in \exptime in combined complexity (even if the bound $k$ and the weights are encoded in binary).
\end{theorem}
\begin{proof}
Let $\WKB$ be a WKB, $k$ an integer and $q = \exists \vec{y} \psi$ a BCQ with $\psi=\bigwedge_{i=1}^n\varphi_i$ where each $\varphi_i$ is an atom of the form $A(t)$ or $R(t,t')$ with $t,t'\in\NI\cup\vec{y}$. 

Let $\mn{N}_{\vec{y}}\subseteq\NI\setminus\inds(\K) $ such that $|\mn{N}_{\vec{y}}|=|\vec{y}|$, and for every valuation $v : \vec{y} \mapsto \inds(\K) \cup \mn{N}_{\vec{y}}$ let $v(\varphi_i)$ denote the fact obtained by replacing each variable $x$ by $v(x)$ in $\varphi_i$ and define a WKB: 
$\K^v_{\omega_v} = (\langle \T , \A_v \rangle ,\omega_v)$ with $\A_v = \A \cup \{ v(\varphi_i)\mid 1\leq i \leq n \}$ and $\omega_v$ extends $\omega$ with $\omega_v(v(\varphi_i))= \infty$ for $1\leq i \leq n$. 
We show that $\wkb \sat{p}{k} q$ iff there exists 
$v$ such that $\K^v_{\omega_v}$ is $k$-satisfiable. 

\noindent($\Leftarrow$) If there exists $v$ such that $\K^v_{\omega_v}$ is $k$-satisfiable, let $\I$ be such that $\cost{\K^v_{\omega_v}}{\I} \leq k$. By construction of $\K^v_{\omega_v}$, $\I\models v(\varphi_i)$ for $1\leq i\leq n$ so $v$ is a match for $q$ in $\I$, \ie $\I\models q$. Moreover, $\cost{\wkb}{\I}=\cost{\K^v_{\omega_v}}{\I}\leq k$. Hence $\wkb \sat{p}{k} q$.

\noindent($\Rightarrow$) If $\wkb \sat{p}{k} q$, there exists $\I\models q$ with $\cost{\wkb}{\I}\leq k$. 
Let $\pi$ be a match for $q$ in $\Imc$ (note that $\pi(c)=c$ for every $c\in\NI$). 
Consider $\mn{D}^\pi_{\vec{y}}:=\{\pi(x)\mid x\in\vec{y}\}\setminus\inds(\K)$. Since $|\mn{D}^\pi_{\vec{y}} |\leq |\mn{N}_{\vec{y}}|$, we can define an injective function $f$ from $\mn{D}^\pi_{\vec{y}}$ to $\mn{N}_{\vec{y}}$. 
Let $v:\vec{y}\mapsto\inds(\K)\cup\mn{N}_{\vec{y}}$ such that $v(x)=\pi(x)$ if $\pi(x)\in \inds(\K)$ and $v(x)=f(\pi(x))$ otherwise, and define $\I_v$ by $\Delta^{\I_v}=\Delta^\I\setminus\mn{D}^\pi_{\vec{y}}\cup \mn{N}_{\vec{y}}$, $c^{\I_v}=c$ for every $c\in\mn{N}_{\vec{y}}$, and for every $A\in\NC$ and $R\in\NR$, substitute $\pi(x)\in\mn{D}^\pi_{\vec{y}}$ with $v(x)$ in $A^\I$ (\resp $R^\I$) to obtain $A^{\I_v}$ (\resp $R^{\I_v}$). 
By construction of $\I_v$, $\I_v\models v(\varphi_i)$ for $1\leq i\leq n$. Moreover, for every $\alpha\in\A$, $\I_v\models \alpha$ iff $\I\models \alpha$ and for every $\tau\in\T$, $\vio{\tau}{\I_v}=\vio{\tau}{\I}\setminus\mn{D}^\pi_{\vec{y}}\cup\{f(e)\mid e\in\vio{\tau}{\I}\cap\mn{D}^\pi_{\vec{y}}\}$. Hence $\cost{\K^v_{\omega_v}}{\I_v}=\cost{\wkb}{\I}\leq k$ and $\K^v_{\omega_v}$ is $k$-satisfiable. 

Therefore, checking for every valuation $v$ (there are at most $(|\inds(\K)| + |q|)^{|q|}$ such valuations) whether $\K^v_{\omega_v}$ is $k$-satisfiable (in exponential time \wrt $|\K^v_{\omega_v}|$ and $|k|$ by Theorem \ref{thm:upperboundBCS}, even with binary encoding of $k$ and the weights) yields an \exptime procedure to decide \dec{C}{b}{p}.

Regarding \dec{C}{opt}{p}, we obtain an \exptime decision procedure by first computing $\optc{\wkb}$ in exponential time \wrt $|\wkb|$ using  Lemma \ref{lem:compOpt}, then applying the \exptime procedure for \dec{C}{b}{p} using $\optc{\wkb}$ as the bound (since by Lemma~\ref{lem:boundOptc} $\optc{\wkb}$ is exponentially bounded in $|\wkb|$, its binary encoding is polynomial in $|\wkb|$). 
\end{proof}
 
 Regarding our variants of the certain semantics, we need to distinguish between IQs and CQs.

\begin{theorem}
	\dec{C}{b}{c} and \dec{C}{opt}{c} for $\mathcal{ALCO}$ are in \twoexptime  in combined complexity and \dec{I}{b}{c} and  \dec{I}{opt}{c} for $\mathcal{ALCO}$ are in \exptime in combined complexity (even if the bound $k$ and the weights are encoded in binary). 
\end{theorem}
\begin{proof}
Let $\WKB$ be a WKB, $k$ an integer and $q$ a BCQ. 
 We have that $\wkb \sat{c}{k} q$ iff $\I \models q$ for every interpretation $\I$ with $\cost{\wkb}{\I} \leq k$. By Lemma \ref{lem:kconfig}, this is the case iff for every $k$-configuration $\gamma$ of $\wkb$, for every  $\I$ satisfying $\gamma$, $\I \models q$. By Proposition \ref{prop:satisfactionConfigKB}, this holds iff for every $k$-configuration $\gamma$ of $\wkb$, for every  $\I\models\K_\gamma$, $\I \models q$. Hence we obtain that $\wkb \sat{c}{k} q$ iff $\K_\gamma\models q$ for every $k$-configuration $\gamma$ of $\wkb$. 
Therefore, checking for every $k$-configuration $\gamma$ for $\wkb$ whether 
$\K_\gamma \models q$  yields a decision procedure for \dec{C}{b}{c}. 

To obtain that \dec{C}{b}{c} is in \twoexptime  in combined complexity and \dec{I}{b}{c} is in \exptime in combined complexity, even if the bound $k$ and the weights are encoded in binary, we use the following facts: 
(i) the number of $k$-configurations is exponentially bounded and each $\K_\gamma$ is of polynomial size and can be constructed in polynomial time (\cf proof of Theorem~\ref{thm:upperboundBCS}), (ii) satisfiability of $\ALCOQu$ is in \exptime even with binary encoding in number restrictions (\cf proof of Theorem~\ref{thm:upperboundBCS}), and (iii) BCQ entailment in tame $\ZOIQ$ (which extends $\ZOQ$, hence $\ALCOQu$) is in \twoexptime,  and in \exptime in the case of BIQ, even with binary encoding in number restrictions \cite[Theorem 8]{ijcai2019p212}. 
Note that the complexity results for tame $\ZOIQ$ apply even if they are shown in a context where the SNA is not made because $\Kmc_\gamma \models q$ under our version of the SNA iff $\Kmc_\gamma\cup \{ \{a\}\sqcap\{b\}\sqsubseteq \bot \mid \!a, b\! \in\! \inds(\K),a\! \neq\! b\}\models q$ without any assumption on the way the individual names are interpreted.

Regarding \dec{C}{opt}{c} (\resp \dec{I}{opt}{c}), we obtain a \twoexptime (\resp \exptime) decision procedure as we did in the proof of Theorem~\ref{th:combined-upper-cqa-p} by first computing $\optc{\wkb}$ in exponential time then applying the procedure for \dec{C}{b}{c} (\resp \dec{I}{b}{c}) using $\optc{\wkb}$ as the bound (using the fact that its binary encoding is polynomial in $|\wkb|$). 
\end{proof}
\subsection{Lower Bounds}

We first prove the hardness of bounded cost satisfiability  for $\mathcal{EL}_\bot$ using a reduction from concept cardinality query answering for $\mathcal{EL}$ KBs  \cite{biem}. 
Given a \emph{concept cardinality query} $q_A$, 
where $A\in\NC$, the answer to $q_A$ in an interpretation $\I$, denoted $q_A^\I$,
 is equal to the cardinality of $A^\I$. 
A certain answer to $q_A$ w.r.t.\ a KB $\K$ is an interval $[m, M] \in (\mathbb{N}\cup\{\infty\})^2$ such that $q_A^\I \in [m, M]$ for every model $\I$ of $\K$.

\begin{theorem}\label{thm:lowerboundBCS}
    \BCS for $\mathcal{EL}_\bot$ is \exptime-hard in combined complexity (even if the bound $k$ and weights are in unary).
\end{theorem}

\begin{proof}
Let $\K = \langle \T , \A \rangle$ be an $\mathcal{EL}$ KB and $q_A$ be a concept cardinality query. 
Define the following $\mathcal{EL}_\bot$ WKB:
\begin{align*}
&\wkb' = (\langle  \T \cup \{ A\sqsubseteq \bot\} , \A \rangle, \omega)\\
&\omega(\chi) = \infty\text{ for } \chi \in \T \cup \A\quad\quad\quad\omega(A\sqsubseteq\bot) = 1
\end{align*}
For every model $\I $ of $\K$, $q_A^\I=|A^\I|=|\vio{A\sqsubseteq\bot}{\I}|$, so $\cost{\wkb'}{\I} = |\vio{A\sqsubseteq\bot}{\I}|=q_A^\I$. 
It follows that $[m,\infty]$ is a certain answer to $q_A$ iff $\cost{\wkb'}{\I} \in [m, \infty]$ for every model $\I$ of $\K$, \ie iff there is no $\I$ with 
$\cost{\wkb'}{\I}< m$.  

As deciding if $[m,\infty]$ is a certain answer to a cardinality query in $\mathcal{EL}$ is \exptime-hard \cite[Theorem 42]{Maniere}, 
BCS for $\ELbot$ is \exptime-hard in combined complexity. Moreover, our reduction only uses weights independent from $|\Kmc|$ and the proof of \cite[Theorem 42]{Maniere} uses an $m$ linear in $|\A|$, 
so BCS \exptime-hardness holds even if the bound $k$ and the weights are encoded in unary.
\end{proof}

We next reduce BCS in $\mathcal{EL}_\bot$ to \dec{I}{b}{p} and \dec{I}{b}{c} for $\mathcal{EL}_\bot$ to leverage this hardness result to IQ (and thus CQ) answering under the $k$-cost-bounded semantics.

\begin{theorem}\label{thm:lowerboundBoundedSem}
	\dec{I}{b}{p} and \dec{I}{b}{c} for $\mathcal{EL}_\bot$ are \exptime-hard in combined complexity (even if  the bound $k$ and the weights are encoded in unary).
\end{theorem}
\begin{proof}
Let $\WKB$ be an $\mathcal{EL}_\bot$ WKB and \emph{k} an integer. 
Let $B\in\NC\setminus\signatureof[\Kmc]$ and $b\in\NI\setminus\inds(\Kmc)$. 

First note that $k < \optc{\wkb}$ 
iff $\wkb \sat{c}{k} B(b)$, due to 
Proposition \ref{k-VariesProp} and the fact that $B$ and $b$ do not occur in $\K$. 
As it is \exptime-hard  to decide whether $k < \optc{\wkb}$ (Theorem \ref{thm:lowerboundBCS}), 
this yields the lower bound for \dec{I}{b}{c}.

Now for \dec{I}{b}{p}, let $\K'_{\omega'} = (\langle \T,\A \cup \{B(b)\} \rangle , \omega' )$ 
where $\omega'$ extends $\omega$ with $\omega'(B(b)) = \infty$. Then $\K'_{\omega'} \sat{p}{k} B(b) $ iff there exists an interpretation $\I$ such that $\cost{\K'_{\omega'}}{\I} \leq k$ and $\I\models B(b)$. For every $\I$ such that $\I\models B(b)$, $\cost{\wkb}{\I}=\cost{\K'_{\omega'}}{\I}$. Moreover, we can add $b^\I$ to $B^\I$ in any interpretation without changing $\cost{\wkb}{\I}$ since $b$ and $B$ do not occur in $\Kmc$. 
Thus $\K'_{\omega'} \sat{p}{k} B(b) $ iff there exists $ \I $ with $\cost{\wkb}{\I} \leq k$. 
\end{proof}

To show the lower bounds for IQ (and thus CQ) answering under the opt-cost semantics, we use a reduction from the problem of deciding if an $\mathcal{EL}$ KB with closed concept names is satisfiable: given a KB and a set of concept names decide if there exists a model $\I$ of the KB such that for every closed concept name $A$, if $d \in A^\I$ then $A(d)$ is in the ABox. 

\begin{restatable}{theorem}{ThIQAoptLowerBoundCombined}\label{th:ThIQAoptLowerBoundCombined}
	\dec{I}{opt}{p} and \dec{I}{opt}{c} for $\mathcal{EL}_\bot$ are \exptime-hard in combined complexity (even if the weights are encoded in unary).
\end{restatable}
\begin{proof}[Proof sketch]
We reduce the \exptime-hard problem of deciding if an $\mathcal{EL}$ KB with closed concept names is satisfiable  \cite{NgoOrSi} to \dec{I}{opt}{p} and \dec{I}{opt}{c}. 
Our reductions are adapted from the proof of the \exptime-hardness of concept cardinality query answering from \cite[Theorem 42]{Maniere}. 
\end{proof}

Finally, we strengthen the lower bounds for CQs and the variants of certain semantics, matching the upper bounds, by adapting a proof for the \twoexptime-hardness of CQ evaluation on circumscribed $\mathcal{EL}$ KBs. A circumscribed KB specifies that some predicates should be minimized, that 
is, the extensions of these predicates in models of the circumscribed KB must be minimal regarding set inclusion.

\begin{restatable}{theorem}{ThCQAcertainLowerBoundCombined}	\label{thm:CQAcertainLowerBoundCombined}
\dec{C}{b}{c} and \dec{C}{opt}{c} for $\mathcal{EL}_\bot$ are \twoexptime-hard in combined complexity (even if  the bound $k$ and the weights are encoded in unary).
\end{restatable}
\begin{proof}[Proof sketch]
We adapt the proof of the \twoexptime-hardness of CQ evaluation on circumscribed $\mathcal{EL}$ KBs from \cite[Theorem 2]{lutz2023querying}, which proceeds by reduction from the \twoexptime-hard problem of BCQ entailment for $\mathcal{EL}$ KBs with closed predicates \cite{NgoOrSi}. 
\end{proof}


\section{Data Complexity}\label{sec:data}
We now turn our attention to the data complexity of the decision problems we consider. 
The second line of Table~\ref{tab:combinedResults} gives an overview of the results. 
Recall that data complexity takes only into account the size of the weighted ABox $|\Amc_\omega|$. 
In particular, for problems that have an integer $k$ as part of their input, we consider that $k$ is fixed. We will discuss the complexity \wrt  $|\Amc_\omega|$ and $k$ at the end of the section.

\subsection{Upper Bounds}

To obtain the data complexity upper bounds, our general approach is to guess a `small' interpretation of bounded or optimal cost, that satisfies or does not satisfy the query. 
It follows from Proposition~\ref{prop:boundedModelALCO} that if $\K_\infty$ is satisfiable,
then $\K_\omega$ admits an interpretation of finite cost whose domain 
is linearly bounded in the size of the ABox.
The following proposition shows that there is always such a bounded-cardinality interpretation which is also of optimal cost.

\begin{restatable}{proposition}{PropBoundedSizekInterpretation}\label{prop:boundedSizekInterpretation}
	Let $\WKB$ be an $\mathcal{ALCO}$ WKB and $k$ an integer. If there exists an interpretation $\I$ with $\cost{\wkb}{\I} \leq k$, then there is $\I'$ with $\cost{\wkb}{\I'} \leq k$ and 
	$|\Delta^{\I'}|\leq |\inds(\K)| + 2^{|\T|}$.
\end{restatable}

The upper data complexity bound for bounded cost satisfiability follows directly.
\begin{theorem}
	\BCS for $\mathcal{ALCO}$ is in \np in data complexity (even with a binary encoding of weights). 
\end{theorem}
\begin{proof}
Let $\WKB$ be an $\mathcal{ALCO}$ WKB and $k$ an integer. 
By Proposition~\ref{prop:boundedSizekInterpretation}, guessing an interpretation $\I$ of cardinality at most $|\inds(\K)| + 2^{|\T|}$ and checking if $\cost{\wkb}{\I} \leq k$ (which can be done in a polynomial time w.r.t $|\A_\omega|$) is an \np procedure to decide BCS.
\end{proof}

To obtain the results for query entailment, we need to refine Proposition~\ref{prop:boundedSizekInterpretation} to preserve also query (non) entailment.

\begin{restatable}{proposition}{PropCQboundedPossible}\label{prop:CQboundedPossible}
	Let $\WKB$ be an $\mathcal{ALCO}$ WKB, $k$ an integer and $q$ a BCQ.  
	If there exists an interpretation $\I$ with $\cost{\wkb}{\I} \leq k$ and $\I \models q$,  then there is one whose domain has cardinality at most $|\inds(\K)| + 2^{|\T|}$.
\end{restatable}

Preserving query non-entailment is more complex. 
Contrary to the proofs of Propositions \ref{prop:boundedSizekInterpretation} and \ref{prop:CQboundedPossible} that build on the notion of filtration used to show the bounded model property of $\mathcal{ALC(O)}$, 
the following proposition relies on a non-trivial adaptation of constructions introduced in the context of counting queries
\cite{biem,Maniere}. The latter work shows how to convert a (potentially infinite) 
interpretation into a finite one while avoiding the introduction of new query matches. 
In our case, we must further prevent new violations of soft TBox axioms. 

\begin{restatable}{proposition}{PropCQboundedCertain}\label{prop:CQboundedCertain}
	Let $\WKB$ be an $\mathcal{ALCO}$ WKB, $k$ an integer, and $q$ a BCQ. If there exists an interpretation~$\I$ with $\cost{\wkb}{\I} \leq k$ and $\I \not \models q$, 
	 then there is one whose domain has cardinality that is bounded polynomially in $|\A|$ and $k$ (with $|\T|$ and $|q|$ treated as constants). 
\end{restatable}

We are now ready to prove the data complexity upper bounds for the $k$-cost-bounded semantics.

\begin{theorem}\label{thm:data-ubCk}
	\dec{C}{b}{p} for $\mathcal{ALCO}$ is in \np in data complexity and \dec{C}{b}{c} for $\mathcal{ALCO}$ is in \conp in data complexity (even if the weights are encoded in binary).
\end{theorem}
\begin{proof}
Let $\WKB$ be an $\mathcal{ALCO}$ WKB, $k$ an integer, and $q$ a BCQ.  
For \dec{C}{b}{p}, we know from Proposition~\ref{prop:CQboundedPossible} that it suffices to guess an interpretation $\I$ whose domain has cardinality at most $|\inds(\K)| + 2^{|\T|}$, 
and check that $\cost{\wkb}{\I} \leq k$ and $\I \models q$ (both checks being possible in polynomial time w.r.t $|\A_\omega|$), yielding an \np procedure. 
The argument is similar for \dec{C}{b}{c}, but uses Proposition~\ref{prop:CQboundedCertain}, which gives the desired polynomial bound in $|\A|$ on the interpretation domain, since $k$ is treated as fixed. 
\end{proof}

For the opt-cost semantics, we use the bound on the optimal cost to compute it by binary search before guessing a `small' interpretation of optimal cost that entails (or does not entail) the query. 
We recall that 
$\Theta^p_2$ is the class of problems which are solvable in polynomial time with at most logarithmically many calls to an \np oracle.

\begin{theorem}
	If there is an ABox-independent bound on the finite weights or weights are encoded in unary, then  \dec{C}{opt}{p} and \dec{C}{opt}{c} for $\mathcal{ALCO}$ are in $\deltaptwolog$ in data complexity.
\end{theorem}
\begin{proof}
Let $\WKB$ be an $\mathcal{ALCO}$ WKB, and $q$ a BCQ. Suppose there is a bound on the maximum finite weights that is independent from $|\A|$, or alternatively, that the weights are encoded in unary. 
Observe that under this assumption, the bound \emph{L}  on the optimal cost given in Lemma~\ref{lem:boundOptc} is polynomial in $|\A_\omega|$. 
	By doing a  binary search using calls to an \np oracle that decides BCS, we can compute the value of 
	$\optc{\wkb}$ with a logarithmic numbers of such calls. 
	It then suffices to make a final call to an \np (resp.\ \conp) oracle that decides \dec{C}{b}{p} (resp.\  \dec{C}{b}{c}) (Theorem~\ref{thm:data-ubCk}) with $k = \optc{\wkb}$.
	Indeed, it is easily seen from the proof of Theorem~\ref{thm:data-ubCk} that the \np / \conp upper bounds hold not only for fixed $k$, but also when $k$ is polynomial in $|\A_\omega|$. 
         We obtain $\deltaptwolog$ procedures for deciding  \dec{C}{opt}{p} and \dec{C}{opt}{c}.
\end{proof}

In the case where the weights are encoded in binary and not bounded independently from the ABox size, we can further show that \optsem{possible} is in $\deltaptwo$ \wrt data complexity (solvable in polynomial time with an \np oracle). 
\begin{theorem}
\dec{C}{opt}{p} for $\mathcal{ALCO}$ is in $\deltaptwo$ in data complexity. 
\end{theorem}
\begin{proof}
Let $\WKB$ be an $\mathcal{ALCO}$ WKB, $q$ a BCQ, and let $L$ be the bound on the optimal cost given in Lemma~\ref{lem:boundOptc}. 
The weights of assertions are encoded in $\Amc_\omega$ so $\max_{\alpha\in\A\setminus\A_\infty}(\omega_\alpha)$ and thus $L$ 
is at most exponential in $|\A_\omega|$ (in the case of binary encoding). 
By doing a binary search using calls to an \np oracle that decides BCS we can compute 
$\optc{\wkb}$ with a numbers of such calls logarithmic in $L$ (so a polynomial \wrt $|\A_\omega|$). Therefore using a final call to an \np oracle that decides \dec{C}{b}{p} (Theorem~\ref{thm:data-ubCk}) with $k = \optc{\wkb}$ we obtain a $\deltaptwo$ procedure to decide \dec{C}{opt}{p}. 
Indeed, even if $k$ might be exponential \wrt  $|\A_\omega|$, the proof of Theorem~\ref{thm:data-ubCk} for the  \bsem{\emph{possible}} relies on guessing an interpretation whose domain cardinality is bounded independently from $k$, so the \np upper bound for \dec{C}{b}{p} holds even if $k$ depends arbitrarily on $|\A_\omega|$. 
\end{proof}

We leave open the question of whether the same upper bound can be obtained for \optsem{certain}. The reason is that  the proof of Theorem~\ref{thm:data-ubCk} for the \bsem{\emph{certain}} relies on guessing an interpretation whose domain cardinality is bounded polynomially in $k$, hence exponentially in $|\A_\omega|$ when $k = \optc{\wkb}$ and the weights are encoded in binary and not bounded independently from $|\A|$.

\subsection{Lower Bounds}

We start by showing that bounded cost satisfiability is \np-hard, using an adaptation of the proof of the \np-hardness of cardinality query answering \cite[Theorem 48]{Maniere}.

\begin{theorem}\label{thm:lbDataBCS}
	BCS for $\mathcal{EL}_\bot$ is \np-hard in data complexity (even with a unary encoding of weights).
\end{theorem}
\begin{proof}
We reduce the 3-colorability problem to \BCS. Let $\mathcal{G} = (\mathcal{V}, \mathcal{E})$ be a graph, and define $\WKB$ with:
\begin{align*}
\T=&\{\exists R.C_i \sqcap \exists E.(\exists R.C_i) \sqsubseteq B \mid 1\leq i\leq 3\}\cup
\\&
\{A \sqsubseteq \exists R.B, B \sqsubseteq \bot\}\\
\A=& \{ A(v) \mid v \in \mathcal{V}\} \cup \{E(v_1,v_2) \mid \{ v_1,v_2 \} \in \mathcal{E} \} \cup
\\& \{C_1(c_1),C_2(c_2),C_3(c_3),B(c_1),B(c_2),B(c_3) \}
\end{align*}
$\omega(B \sqsubseteq \bot)=1$ and $\omega(\chi)=\infty$ for all other $\chi\in\T\cup\A$.

We show that $\wkb$ is 3-satisfiable iff $\mathcal{G}$ is 3-colorable.

\noindent($\Leftarrow$) If $\mathcal{G}$ is 3-colorable, let $\rho: \mathcal{V} \mapsto \{c_1,c_2,c_3\}$ be a 3-coloring of $\mathcal{G}$ and let $\I_\rho$ be the interpretation that satisfies exactly the assertions of $\A\cup\{R(v, \rho(v))\mid v \in \mathcal{V}\}$.  We show that $\I_\rho\models \K_\infty$, so that $\cost{\wkb}{\I_\rho} = |B^{\I_\rho}|= 3$. 
By construction, $\I_\rho\models\Amc$ and 
since $A^{\I_\rho}=\mathcal{V}$, the $R(v, \rho(v))$ assertions ensure that $\I_\rho\models A \sqsubseteq \exists R.B$. Moreover, by definition of $\rho$, there is no monochromatic edge, which ensures that $(\exists R.C_i \sqcap \exists E.(\exists R.C_i))^{\I_\rho}=\emptyset$ for $i\in\{1,2,3\}$.

\noindent($\Rightarrow$) If $\wkb$ is 3-satisfiable, since $\{c_1,c_2,c_3\}\subseteq B^\I$ for any $\I$ of finite cost, then there exists $\I$ such that $\cost{\wkb}{\I} = 3$ and $B^\I=\{c_1,c_2,c_3\}$. 
For every $v\in \mathcal{V}$, since $\I\models A(v)$ and $\I\models A\sqsubseteq\exists R.B$, then there exists $c\in \{c_1,c_2,c_3\}$ such that $(v,c)\in R^\I$. 
Define a coloring $\rho: \mathcal{V} \mapsto \{c_1,c_2,c_3\}$ by arbitrarily selecting one such $c$ per $v$:  $(v,\rho(v))\in R^\I$ for each $v\in\mathcal{V}$. We show that $\rho$ is a 3-coloring of $\mathcal{G}$. Otherwise, if there was an edge $(v_1,v_2)$ such that $\rho(v_1)=\rho(v_2)=c_i$, 
then $\I\models E(v_1,v_2)\wedge R(v_1,c_i)\wedge R(v_2,c_i) \wedge C_i(c_i)$, \ie  $v_1\in (\exists R.C_i \sqcap \exists E.(\exists R.C_i))^\I$. It would follow that $v_1\in B^\I$ and 
$\cost{\wkb}{\I} \geq 4$, contradicting $\cost{\wkb}{\I} = 3$.
\end{proof}

A direct adaptation of the last proof gives \np and \conp lower bounds for IQ (and thus CQ) answering under the $k$-cost-bounded semantics.

\begin{theorem}
	\dec{I}{b}{p} (\resp \dec{I}{b}{c}) for $\mathcal{EL}_\bot$ is \np-hard (\resp \conp-hard) in data complexity (even with a unary encoding of weights).
\end{theorem}
\begin{proof}
Given a graph $\mathcal{G}$, consider the WKB $\wkb$ defined in the proof of Theorem \ref{thm:lbDataBCS} and let $D\in\NC\setminus\signatureof[\Kmc]$ and $d\in\NI\setminus\inds(\Kmc)$. 
Then $\mathcal{G}$ is \emph{not} 3-colorable iff $\wkb$ is \emph{not} $3$-satisfiable, iff $\wkb \sat{c}{3} D(d)$, yielding the lower bound for \dec{I}{b}{c}. 
Now for \dec{I}{b}{p}, let $\K'_{\omega'} = (\langle \T,\A \cup \{D(d)\} \rangle , \omega' )$ where  $\omega'$ extends $\omega$ with $\omega'(D(d)) = \infty$. Then $\K'_{\omega'} \sat{p}{3} D(d) $ iff 
there exists $ \I $ with $\cost{\wkb}{\I} \leq 3$, iff $\mathcal{G}$ is 3-colorable. 
\end{proof}

The reduction used for the next theorem is inspired by the proof of $\deltaptwolog$-hardness of the $\leq$-AR semantics (based upon cardinality-maximal repairs, or equivalently, weight-based ABox repairs with all assertions assigned equal weight) 
for some existential rule languages \cite[Theorem 8]{Lukasiewicz_Malizia_Vaicenavicius_2019}. 

\begin{restatable}{theorem}{ThIQAoptDataLowerBound}
	\dec{I}{opt}{p} and \dec{I}{opt}{c} for $\mathcal{EL}_\bot$ are $\deltaptwolog$-hard in data complexity (even if the finite weights are bounded independently from $|\A|$ and encoded in unary).
\end{restatable}
\begin{proof}[Proof sketch]
We use a reduction from the problem of deciding whether a given vertex belongs to all the independent sets of maximum size of a graph. 
\end{proof}

When the assertions weights can be exponential \wrt $|\A|$, one can encode priority between the assertions, which we use to prove the following result. 

\begin{restatable}{theorem}{ThCQAoptDataLowerBound}
	\dec{I}{opt}{p} and \dec{I}{opt}{c} for $\mathcal{EL}_\bot$ are $\deltaptwo$-hard in data complexity (if the weights are encoded in binary).
\end{restatable}
\begin{proof}[Proof sketch]
The result is proved by reduction from the following $\deltaptwo$-hard problem:
given a satisfiable Boolean formula $\varphi = c_1 \wedge \ldots \wedge c_m$ over variables $x_1,\dots, x_n$ such that each clause $c_i$ 
has exactly two positive and two negative literals (with \texttt{true} and \texttt{false} admitted as literals) 
and given  $k\in\{1,\dots, n\}$, decide whether the lexicographically maximum truth assignment $\nu_{max}$ satisfying $\varphi$ \wrt $(x_1,\dots, x_n)$ fulfills $\nu_{max}(x_k)=1$. 
 It follows from \cite{Krentel88} and from the reductions from SAT to 3SAT and from 3SAT to 2+2SAT \cite{Schaerf93} that this problem is $\deltaptwo$-hard.

Let $\varphi$ be an instance of the problem as previously defined. We define  an \ELbot WKB $\WKB$ as follows.
\begin{align*}
\A=&\{S(\varphi,c_j)\mid 1\leq j\leq m\}\cup\\
&\{P_\ell(c_k,x_j) \mid \ell\in\{1, 2\},x_j\text{ is the }\ell^{th}\text{ pos. lit. of }c_k \}\cup\\
    &\{N_\ell(c_k,x_j) \mid \ell\in\{1, 2\},\neg x_j\text{ is the }\ell^{th}\text{ neg. lit. of }c_k \}\cup\\
    & \{F(x_i), T(x_i), T'(x_i)\mid 1\leq i\leq n\}\cup\{F(f), T(t)\}
    \\
\T=&\{F\sqcap T\sqsubseteq \bot, F\sqcap T'\sqsubseteq \bot\}\cup\\
&\{  \exists S.(  \exists P_1.F\sqcap\exists P_2.F\sqcap \exists N_1.T\sqcap \exists N_2.T) \sqsubseteq \bot \}
\end{align*}
We set $\omega(\tau)=\infty$ for every $\tau\in\T$, and $\omega(\alpha)=\infty$ for every  assertion $\alpha$ built on $P_\ell$ and $N_\ell$ as well as for $T(t)$ and $F(f)$. For the remaining assertions, we define the weights through the following prioritization \cite[Lemma~6.2.5]{DBLP:phd/hal/Bourgaux16}:
\begin{itemize}
\item $L_1=\{T(x_i), F(x_i)\mid 1\leq i\leq n\}$
\item $L_2=\{S(\varphi,c_j)\mid 1\leq j\leq m\}$
\item $L_3=\{T'(x_1)\}$, $L_4=\{T'(x_2)\}$, $\dots$, $L_{n+2}=\{T'(x_n)\}$
\end{itemize} 
Let $u=max(2n,m)+1$ and define $\omega(\alpha)=u^{n+2-i}$ for every $\alpha\in L_i$. 
We show that $\wkb \sat{c}{opt} T'(x_k)$ iff $\wkb \sat{p}{opt} T'(x_k)$ iff $\nu_{max}(x_k)=1$.
\end{proof}

Interestingly, while the preceding proof crucially relies upon a binary encoding of exponential weights, all of our other lower 
bounds only employ 1 and $\infty$ as weights.

\subsection{Taking $k$ Into Account}

We defined the data complexity as the complexity \wrt $|\A_\omega|$ but it is interesting to investigate the impact of considering $k$ as part of the input. 
Naturally, all data complexity lower bounds hold when $k$ is also part of the input. 
Moreover, it can be checked that the data complexity upper bounds also hold for the complexity \wrt $|\A_\omega|$ and $|k|$, \emph{except 
for the $k$-cost-bounded certain semantics} (Theorem \ref{thm:data-ubCk}). The difficulty comes from the proof of Proposition \ref{prop:CQboundedCertain}: the bound obtained for the cardinality of the domain of the new interpretation is polynomial in $k$ thus becomes exponential in $|k|$ if $k$ is encoded in binary. Therefore, if $|k|$ is treated as part of the input, the proof of Theorem \ref{thm:data-ubCk} only yields a \conp upper bound for \dec{C}{b}{c} if we assume a \emph{unary encoding of~$k$}.

\section{Related Work}
Our cost-based framework for reasoning on inconsistent KBs shares features
with a number of existing formalisms. 
Our axioms with finite weights  
can be seen as quantitative versions of the soft rules considered in \cite{eit}.
More generally, the idea of allowing TBox axioms to be (exceptionally) violated 
shares high-level similarities with defeasible axioms, minimized concepts, and typicality operators  
considered in non-monotonic extensions of DLs, cf.\ \cite{DBLP:journals/jair/BonattiLW09,DBLP:journals/ai/GiordanoGOP13,DBLP:journals/tocl/BritzCMMSV21}.
Although we differ in adopting  
a quantitative 
semantics for axiom violations, 
we were nevertheless able to 
import some   
techniques from  
circumscription \cite{lutz2023querying} 
and reasoning with closed predicates \cite{NgoOrSi}.

Our approach is related to existing quantitative notions of repair. Indeed, our  
assignment of costs to interpretations 
can be seen as adapting the database repairs for soft constraints from  \cite{DBLP:conf/icdt/CarmeliGKLT21}
to the DL setting. 
As detailed in Section \ref{weighted knowledge bases}, 
our approach is also 
related to weight-based ABox repairs:  
when TBox axioms have weight $\infty$, 
our \optsem{certain} coincides with the $\leq_\omega$-AR semantics from \cite{biebougoa}.

Our proofs rely upon techniques devised for other forms of quantitative reasoning in DLs. 
We exploit results on counting and cardinality queries \cite{biem,Maniere}
and reduce some of our problems to reasoning 
in DLs with number restrictions. 
We could also have obtained some of our combined complexity upper bounds 
by using results on \emph{cardinality constraints} $\leq n C$ (which enforce that $|C^\Imc|\leq n$), 
previously studied in \cite{DBLP:journals/jair/Tobies00} and expressible in the 
much more expressive quantitative logics in \cite{baader2020satisfiability}. 
We expect that the techniques from these works 
may prove useful in future studies of WKBs.

\section{Conclusion}
We have introduced a new cost-based framework to querying inconsistent DL KBs,
in which both TBox axioms and ABox assertions may be violated and notions of 
possible and certain query answers
are defined  w.r.t.\ the bounded or minimal-cost interpretations. 
By exploiting connections to other DL reasoning tasks, 
we were able to establish 
an almost complete picture of the complexity of query answering in our
framework, for DLs between $\ELbot$ and $\ALCO$. 

The main aim of future work will be to extend our results to cover DLs involving other common constructs, like
 inverse roles ($\mathcal{I}$), role inclusions ($\mathcal{H}$), functional roles ($\mathcal{F}$), and number restrictions ($\mathcal{Q}$).  
We will focus first on the DL-Lite family due to its widespread use in OMQA. 
While it is not difficult to adapt the definitions of violations and WKBs (although there may be several ways to count violations of a functionality axiom, as one could count as a single violation a domain element that has multiple successors or instead count every pair of two distinct successors), 
the techniques underlying our upper bounds do not readily transfer. 
For example, violations of role inclusions cannot be encoded as concepts (at least in standard DLs), 
and adding functionality or number restrictions may lead to 
every optimal-cost interpretation having an infinite domain. 
Inverse roles (present in even the simplest DL-Lite logics) 
are also non-trivial to handle, 
though we expect that 
the techniques 
employed for our data complexity upper bounds can be suitably adapted. Finally, another important but challenging direction is to devise practical algorithms that are amenable to implementation.

\section*{Acknowledgements} 
We thank the reviewers for their careful reading and useful feedback. 
This work was partially supported by the ANR AI Chair INTENDED (ANR-19-CHIA-0014) and the ANR project CQFD (ANR-18-CE23-0003).

\bibliographystyle{kr}
\bibliography{reference}

\newpage

\allowdisplaybreaks

\appendix
\section{Proofs for Section \ref{weighted knowledge bases}}

\consistentCase*
\begin{proof}
For every interpretation $\I$ we have $\cost{\wkb}{\I} = 0$ iff $\I$ is a model of $\K$. Indeed, we have $\cost{\wkb}{\I} = \sum_{\tau \in \T} \omega_\tau|\vio{\tau}{\I}| + \sum_{\alpha \in \vio{\A}{\I}}\omega_\alpha$ thus $\cost{\wkb}{\I} = 0$ iff $\vio{\tau}{\I} = \emptyset$ for every $\tau \in \T$ (as $\omega_\tau > 0$ for every $\tau \in \T$) and $\vio{\A}{\I} = \emptyset$  (as $\omega_\alpha > 0$ for every $\alpha \in \A$) iff $\I$ is a model of $\K$.
The result directly follows from this 
characterization of interpretations of cost~$0$ and the definition of the opt-cost certain and possible semantics.
\end{proof}

\kVariesProp*
\begin{proof}
\begin{itemize}
\item Let $0 \leq k' \leq k$ and suppose $\wkb \sat{c}{k} q$ \ie $\I \models q$ for every interpretation $\I$ with $\cost{\wkb}{\I} \leq k$. Then $\I \models q$ for every interpretation $\I$ with $\cost{\wkb}{\I} \leq k' \leq k$, thus $\wkb \sat{c}{k'} q$.
\item Let $0 \leq k' \leq k$ and suppose $\wkb \not \sat{p}{k} q$ \ie $\I \not \models q$ for every interpretation $\I$ with $\cost{\wkb}{\I} \leq k$. Then $\I \not \models q$ for every interpretation $\I$ with $\cost{\wkb}{\I} \leq k' \leq k$, thus $\wkb \not \sat{p}{k'} q$.
\end{itemize}

\noindent 
The last part of the proposition follows from the fact that if $k < \optc{\wkb}$ then there is no interpretation $\I$ with $\cost{\wkb}{\I} = k$. Thus, no query (\resp every query) is entailed under \bsem{possible} (\resp \bsem{certain}).
\end{proof}

\RelPreferredAR*
\begin{proof}
Since $\T$ is satisfiable, there is a model $\I$ of $\T$, and since $\omega(\alpha)\neq \infty$ for every $\alpha\in\A$, $\cost{\wkb}{\I}\neq \infty$. It follows that $\optc{\wkb} \neq \infty$ and that every $\I$ of optimal cost is such that $\I\models \T$ and $\cost{\wkb}{\I}= \sum_{\alpha \in vio_\A(\I)} \omega_\alpha$. 

Hence, for every $\I$ of optimal cost, $\A' = \A \setminus vio_\A(\I)$ is a $\leq_\omega$-repair. Indeed, since $\I\models\langle\T,\A'\rangle$, $\langle\T,\A'\rangle$ is consistent. Moreover, $\sum_{\alpha \in \A'} \omega_\alpha = \sum_{\alpha \in \A} \omega_\alpha - \cost{\wkb}{\I}$ and $\cost{\wkb}{\I}$ is minimal, so $\sum_{\alpha \in \A'} \omega_\alpha$  is maximal.  
It also follows that every $\leq_\omega$-repair $\A'$ is such that $\sum_{\alpha \in \A'} \omega_\alpha = \sum_{\alpha \in \A} \omega_\alpha - \optc{\wkb}$. 

\noindent($\Leftarrow$) Assume that $\wkb \not\sat{c}{opt} q$. There exists $\I$ of optimal cost such that $\I\not\models q$. 
Let $\A' = \A \setminus vio_\A(\I)$ be the $\leq_\omega$-repair associated to $\I$. Since $\I\not\models q$, it follows that $\langle \T, \A'\rangle \not\models q$. Therefore $\langle \T, \A \rangle \not\models_{\leq_\omega\text{-AR}} q$.

\noindent($\Rightarrow$) Assume that $\wkb \sat{c}{opt} q$.     
Let $\A'$ be a $\leq_\omega$-repair and $\I$ be a model of $\T$ and $\A'$. 
Since $\cost{\wkb}{\I}= \sum_{\alpha \in \A\setminus\A'} \omega_\alpha=\sum_{\alpha \in \A} \omega_\alpha -\sum_{\alpha \in \A'} \omega_\alpha$ and $\sum_{\alpha \in \A'} \omega_\alpha = \sum_{\alpha \in \A} \omega_\alpha - \optc{\wkb}$ because $\A'$ is a $\leq_\omega$-repair, 
it follows that $\cost{\wkb}{\I}= \optc{\wkb}$. 
Therefore $\I\models q$. Hence $\langle\T,\A'\rangle\models q$. 
\end{proof}

\section{Proofs for Section \ref{sec:combined}}

\subsection{Proof of Proposition~\ref{prop:boundedModelALCO}}\label{app-filtration}
The proof of Proposition~\ref{prop:boundedModelALCO} is adapted from the proof of $\mathcal{ALC}$ bounded model property by \citeauthor{Baader_Horrocks_Lutz_Sattler_2017}~\shortcite{Baader_Horrocks_Lutz_Sattler_2017}. It is based on the notion of filtration that `merges' elements that belong to the same concepts and that we extend below to handle nominals (while respecting the SNA).

  Let $S$ be a set of $\mathcal{ALCO}$ concepts and $\I$ an interpretation. The \emph{$S$-type} of $d \in \Delta^{\I}$ is defined as $$t_{S}(d)=\{C \in S \mid d \in C^{\I}\}.$$ If $S$ is finite it can be shown that $|\{t_{S}(d) \mid d \in \Delta^{\I} \} | \leq 2^{|S|}$.

Let $\KB$ be an $\mathcal{ALCO}$ KB, $S = sub(\T)$ be the set of all subconcepts of the concepts that occur in the concept inclusions of $\T$, and let $\I$ be an interpretation. The equivalence relation $\simeq_{S}$ on $\Delta^{\I} \setminus \inds(\K)$ is defined as follows:
$$
d \simeq_{S} e \text { if } t_{S}(d)=t_{S}(e) .
$$
The $\simeq_{S}$-equivalence class of $d \in \Delta^{\I}\setminus \inds(\K)$ is denoted by $[d]_{S}$, i.e.,
$$
[d]_{S}=\{e \in \Delta^{\I}\setminus \inds(\K) \mid d \simeq_{S} e\}
$$
The \emph{$S$-filtration} of $\I$ is the following interpretation $\J$ :

\begin{align*}
\Delta^{\J}= & \{[d]_{S} \mid d \in \Delta^{\I}\setminus \inds(\K)\} \cup \inds(\K) \\
A^{\J}= & \{[d]_{S} \mid \text{ there is } d' \in[d]_{S} \text { with } d' \in A^{\I}\}
	\\ &\cup \{ a \in \inds(\K) \mid a \in A^\I \}\text { for all } A \in \NC \\
\{a\}^\J = &  \{ [a^\I]_{S} \}  \text { for all } a \in \NI\setminus \inds(\K) \\ 
\{a\}^\J = &\{a\} \text{ for all } a \in \inds(\K) \\ 
R^{\J}= & \{([d]_{S},[e]_{S}) \mid (d', e') \in R^{\I}, d' \in[d]_{S}, e' \in[e]_{S}\} \\
&\cup \{(a,[d]_{S}) \mid (a, d') \in R^{\I}, a \in \inds(\K), d' \in[d]_{S}\} \\
&\cup \{([d]_{S},a) \mid   (d', a) \in R^{\I}, a \in \inds(\K), d' \in[d]_{S}\} \\
&\cup \{(a,b) \mid (a, b) \in R^{\I}, a,b \in \inds(\K)\} \\
& \text { for all } R \in \NR. 
\end{align*}

It can be shown as in \cite[Lemma 3.15]{Baader_Horrocks_Lutz_Sattler_2017} that the $S$-filtration $\J$ has the same `form' as the interpretation $\I$ we started with in the following sense.

\begin{lemma}\label{lem:mainEquiv}

Let $\KB$ be an $\mathcal{ALCO}$ KB, $S = sub(\T)$, $\I$ an interpretation and $\J$ the $S$-filtration of $\I$. Then for every $C \in S$, 
\begin{itemize}
\item$
d \in C^{\I} \quad \text { iff } \quad  [d]_{S} \in C^{\J} \text{ for all } d \in \Delta^\I \setminus \inds(\K)
$
\item $a \in C^{\I} \quad \text { iff } \quad  a \in C^{\J} \text{ for all } a \in \inds(\K)
$
\end{itemize}

\end{lemma}

Proposition~\ref{prop:boundedModelALCO} follows easily from Lemma~\ref{lem:mainEquiv}. Indeed, if $\K$ has a model $\I'$, we can show that the $S$-filtration $\I$ of $\I'$ with $S = sub(\T)$ is  a model of $\K$ such that $|\Delta^\I| \leq |\inds(\K)| + 2^{|\T|}$.

\PropBoundedModelALCO*

\subsection{Lower Bounds}

\ThIQAoptLowerBoundCombined*

\begin{proof}
We proceed by reduction from the \exptime-hard problem of deciding whether an $\mathcal{EL}$ KB with closed concept names is satisfiable \cite{NgoOrSi} (note that \citeauthor{NgoOrSi}  consider a more general problem where role names can also be closed but as mentioned in their conclusion, the reductions they use to prove complexity lower bounds only require closed concept names). 

An $\mathcal{EL}$ KB with closed concept names takes the form $(\K,\Sigma)$ with $\K$ an $\mathcal{EL}$ KB and $\Sigma \subseteq \NC$ a set of closed concept names. An interpretation $\I$ \emph{respects $\Sigma$} if for every $A \in \Sigma$, $d \in A^\I$ implies $A(d) \in \A$, and $\I$ is a model of $(\K,\Sigma)$ if $\I\models\K$ and $\I$ respects $\Sigma$. 
Let $(\K', \Sigma)= (\langle \T' , \A' \rangle, \Sigma)$ be an $\mathcal{EL}$ KB with closed concept names, which we assume \wlg to be in normal form. 

\smallskip

We start by the reduction for \dec{I}{opt}{c}. Adapting the proof of Theorem 42 in \cite{Maniere}, we define $\tilde{\K}$ as follows. 
Let $Goal_1, Goal_2$ and $Aux_\top$ be fresh concept names, $\{R_B\mid B \in \Sigma\}$ be fresh role names and $aux$ be a fresh individual name. Then $\tilde{\K}=\langle\tilde{\T},\tilde{\A}\rangle$ with 
\begin{align*}
	\tilde{\T} = &\,\T' \setminus \{\top \sqsubseteq A \mid \top \sqsubseteq A \in \T' \}    
	\\ &\cup \{ Aux_\top \sqsubseteq A \mid\top \sqsubseteq A \in \T' \}  
	\\ &\cup \{ A \sqsubseteq Aux_\top \mid A \in \signatureof[\T'] \} 
	\\ &\cup \{ B \sqsubseteq Goal_1 \mid B \in \Sigma \}  
	\\ &\cup \{ \exists R_B.B \sqsubseteq Goal_1 \mid B \in \Sigma \} 
	\\ &\cup \{ Goal_1 \sqsubseteq Goal_2 \}
	\\ \tilde{\A} =&\, \A'
	\cup \{Aux_\top(a) \mid a \in \inds(\A') \}	
	\\ &\cup \{ R_B(aux,a) \mid a \in \inds(\A'), B \in \Sigma, B(a) \not\in \A' \}
	\\ &\cup \{ Goal_1(a) \mid a \in \inds(\A'), B \in \Sigma, B(a) \in \A' \}
\end{align*}	
We then define the $\mathcal{EL}_\bot$ WKB $\WKB$ with 
\begin{align*}
	\T =\, & \tilde{\T} \cup \{Goal_1 \sqsubseteq \bot, Goal_2 \sqsubseteq \bot \} \\
	\A =\, &\tilde{\A} \cup \{ Goal_1(aux) \} \\
	\omega(\tau) =\, & \infty \text{ for } \tau \in \tilde{\T} \\ 
	\omega(\alpha) =\, & \infty \text{ for } \alpha \in \tilde{\A} \\
	\omega(Goal_i \sqsubseteq \bot) =\, & 1 \text{ for } i \in\{ 1,2\} \\
    \omega(Goal_1(aux)) =\, &1
\end{align*} 
We show that $(\K',\Sigma)$ is satisfiable iff $ \wkb\! \not\sat{c}{opt}\! Goal_1(aux)$. 
Let $n := |\{ b \in \NI | B(b) \in \A , B \in \Sigma\}|$.

\noindent($\Rightarrow$) Suppose that $(\K', \Sigma) $ is satisfiable. 
First, note that for every interpretation $\I$, $\cost{\wkb}{\I} \geq 2n+1$. Indeed:
\begin{itemize}
\item if $\I \not\models \Kinf$, then $\cost{\wkb}{\I}=\infty$;
\item if $\I \models \Kinf$, then $\{ b \in \NI | B(b) \in \A , B \in \Sigma\} \subseteq Goal_i^\I$ for $i \in\{ 1,2\}$ so 
\begin{itemize}
\item either $\I \models Goal_1(aux)$ and $\{ b \in \NI | B(b) \in \A , B \in \Sigma\} \cup \{aux\} \subseteq Goal_i^\I$ for $i \in\{ 1,2\}$ so $|Goal_i^\I| \geq n+1$ for $i \in\{ 1,2\}$ i.e. $\cost{\wkb}{\I} \geq 2(n+1)$, 
\item or $\I \not\models Goal_1(aux)$ and $\cost{\wkb}{\I} \geq 2n+1$. 
\end{itemize}	
\end{itemize}	
Let $\J$ be a model of $(\K',\Sigma)$ and define a new interpretation $\J'=(\Delta^\J\cup\{aux\},\cdot^{\J'})$ as follows:
\begin{align*}
	A^{\J'} &:= A^\J		\text{ for } A \in \signatureof[\T'] \cap \NC
	\\ Goal_i^{\J'} &:= \{a \mid B(a) \in \A', B \in \Sigma \} \text{ for } i \in\{1,2\}& 
	\\ Aux_\top^{\J'} &:= \Delta^\J&
	\\  P^{\J'}  &:= P^\J   \text{ for } P \in \signatureof[\T']\cap \NR
	\\  R_B^{\J'} &:= \{(aux,a) \mid a \in \inds(\A'),  B(a) \not\in \A', B \in \Sigma\}
\end{align*}

By an easy adaptation of the proof of \cite[Theorem 42]{Maniere}, $\J'$ is a model of $\tilde{\K} = \Kinf$. Hence $\cost{\wkb}{\J'} = 2n + 1 = \optc{\wkb}$. Since $\J' \not\models Goal_1(aux)$, it follows that $\wkb \not\sat{c}{opt} Goal_1(aux)$. \medskip

\noindent($\Leftarrow$) 
Suppose $\wkb \not\sat{c}{opt} Goal_1(aux)$.
Since $\K'$ is satisfiable (as it is an $\mathcal{EL}$ KB), $\Kinf$ is satisfiable so $\optc{\wkb} < \infty$.
Let $\I$ be an interpretation with optimal cost such that $\I \not\models Goal_1(aux)$ (and hence also $\I \not\models Goal_2(aux)$). Consider the interpretation $\J$ obtained by restricting $\I$ to the domain $\Delta^\J = (Aux_\top)^\I$. We show that $\J$ is a model of $(\K', \Sigma)$.

Since $\cost{\wkb}{\I} < \infty$, $\I \models \Kinf$ thus $\I \models \tilde{\K}$. As in the proof of \cite[Theorem 42]{Maniere}, it follows that  $\J\models \K'$. 
It remains to verify that $\J$ respects $\Sigma$. To this end, let $B \in \Sigma$ and $e \in B^\J$. We aim to show that $B(e)\in\A'$.

 We first show that $e\in \inds(\A')$. Since $B^\J=B^\I\cap \Delta^\J$, we have $e\in B^\I$. Since $\I\models\Tinf$, we obtain $e \in Goal_1^\I$ (hence also $e \in Goal_2^\I$). 
As $\I$ has optimal cost, 
$Goal_1^\I \subseteq \inds(\A') $. 
Indeed, suppose there exists $u \in Goal_1^\I$ such that $u \not \in \inds(\A')$. Consider the interpretation $\I'$ which is defined as $\I$ but where every occurence of $u$ is replaced by $aux$. Formally, we set:
\begin{align*} 
\Delta^{\I'} &= \Delta^\I \\
A^{\I'} &= A^\I \text{ for every } A \in N_C \text{ such that } u \not \in A^\I \\
A^{\I'} &= \{ A^\I \setminus \{ u \} \} \cup \{aux \} \text{ for every } A \in N_C \\
		&\text{ such that } u \in A^\I \\
R^{\I'} &= \{ (b,c) \mid b,c \neq u \text{ and } (b,c) \in R^\I \} \\
		&\cup \{ (aux,c) \mid c \neq u \text{ and } (u,c) \in R^\I \} \\
		&\cup \{ (b,aux) \mid b \neq u \text{ and } (b,u) \in R^\I \} \\
		&\cup \{ (aux,aux) \mid (u,u) \in R^\I \}
\end{align*}
It follows from the construction of $\I'$ that $\I' \models\Kinf$ and  $|Goal_i^{\I'}| = |Goal_i^\I|$ for $i = 1,2$ (since $\I \not\models Goal_i(aux)$). Moreover, $\I'$ does not violate $Goal_1(aux)$. We therefore have  
 $\cost{\K_\omega}{\I'} = |Goal_1^{\I'}| + |Goal_2^{\I'}| < |Goal_1^{\I}| + |Goal_2^{\I}| + 1 = \cost{\K_\omega}{\I}$. 
Therefore $e \in \inds(\A')$.

We now show that $B(e) \in \A'$. 
By construction of $\tilde{\T}$, $B\sqsubseteq Goal_1$ and $\exists R_B.B \sqsubseteq Goal_1$ are in  $\tilde{\T} \subseteq \Tinf$. 
Hence, since $\I \not \models Goal_1(aux)$ and $\I\models\Tinf$, it follows that $aux \notin (\exists R_B.B)^\I$. 
Since $B^\J=B^\I\cap \Delta^\J$, $e\in B^\I$ so $aux \notin (\exists R_B.B)^\I$ implies in particular that $(aux, e)\notin R_B^\I$. 
Since $\I\models\Ainf$, $\I\models R_B(aux, b)$ for every $b\in\inds(\A')$ such that $B(b) \not\in \A'$ so $e\in\inds(\A')$ and $(aux, e)\notin R_B^\I$ means that $B(e)\in\A'$.

We have thus showed that $\J$ respects $\Sigma$, hence that  $(\K', \Sigma) $ is satisfiable.

\medskip

We now present the reduction for \dec{I}{opt}{p}. Let $B$ be a fresh concept name. From the WKB $\WKB$ we used for the reduction for \dec{I}{opt}{c}, we define an $\mathcal{EL}_\bot$ WKB $\K''=(\langle\T'',\A''\rangle,\omega'')$ by setting:
\begin{align*}
\T'' &= \T \cup \{ B \sqcap Goal_1 \sqsubseteq \bot \} \\
\A'' &= \A 
\end{align*}
and letting $\omega''$ extend $\omega$ with $\omega''(B \sqcap Goal_1 \sqsubseteq \bot) = \infty$.

We show that $\K''_{\omega''} \!\sat{p}{opt}\! B(aux)$ iff $\wkb\! \not\sat{c}{opt}\! Goal_1(aux)$ (which we know holds iff $(\K',\Sigma)$ is satisfiable). First, note that for each interpretation $\I$ we have	$\cost{\K''_{\omega''}}{\I} = \cost{\wkb}{\I} + |vio_{B \sqcap Goal_1 \sqsubseteq \bot}(\I)|\omega_{B \sqcap Goal_1 \sqsubseteq \bot}$. Therefore it is clear that $\optc{\wkb} = \optc{\K''_{\omega''}}$ as $B$ is a fresh concept name.

\noindent($\Leftarrow$) Let $\I$ be an interpretation with optimal cost for $\wkb$ such that $\I \not\models Goal_1(aux)$. Let $\J$ be obtained from $\I$ by setting 
$B^\J = \{ aux \}$. We have $\cost{\K''_{\omega''}}{\J} = \cost{\wkb}{\I} = \optc{\wkb} = \optc{\K''_{\omega''}}$ thus $\J$ has optimal cost for $\K''$ and $\J \models B(aux)$. Hence $\K''_{\omega''} \sat{p}{opt} B(aux)$.

\noindent($\Rightarrow$) Let $\I$ be an interpretation with optimal cost for $\K''_{\omega''}$ such that $\I \models B(aux)$. Then $\I$ also has optimal cost for $\wkb$ and $\I \not\models Goal_1(aux)$ (otherwise its cost would be infinite).
\smallskip

Note that in the construction of $\wkb$ and $\K''_{\omega''} $ the weights used are either $1$ or $\infty$ thus the complexity lower bounds hold even for unary encoding of the weights.
\end{proof}

\ThCQAcertainLowerBoundCombined*

\begin{proof} 
We reduce the problem of BCQ entailment on $\mathcal{EL}$ KBs with closed concept names which is known to be \twoexptime-hard \cite{NgoOrSi} (\cf proof of Theorem~\ref{th:ThIQAoptLowerBoundCombined} for the definition of KBs with closed concept names) to \dec{C}{opt}{c}. 
Let $(\tilde{\K},\Sigma)$ be an $\mathcal{EL}$ KB with closed concept names with $\tilde{\K}=\langle\tilde{\T},\tilde{\A}\rangle$ in normal form, and let $q = \exists \vec{y} \phi$ be a BCQ and $var(q)$ be the variables that occur in $q$. 

We adapt the proof of \twoexptime-hardness of CQ evaluation on circumscribed KBs with a single minimized concept name and no fixed concept \cite[Theorem 2]{lutz2023querying}. Define, as in their proof, the following $\mathcal{EL}$ KB $\K'=\langle \T',\A'\rangle$. 
Let $M$, $L$, $X$, $\bar{X}$ and $\bar{A}$ for every $A \in \Sigma$ be fresh concept names, $u$ be a fresh role name, and $t$ be a fresh individual name.
\begin{align*}	
				\T' &= \{ A \sqsubseteq M \mid A \in \Sigma\} 
				\\ & \cup \{ A \sqcap \bar{A} \sqsubseteq L \mid A\in \Sigma\} 
				\\ & \cup \{ \exists u.L \sqsubseteq X \}
				\\ & \cup \{ X \sqsubseteq A \mid \top \sqsubseteq A \in \tilde{\T} \}
				\\ & \cup  \{  X  \sqcap A \sqsubseteq B \mid A \sqsubseteq B \in \tilde{\T} \}
				\\ & \cup \{ X  \sqcap A \sqsubseteq \exists r.(X \sqcap B)\mid A \sqsubseteq \exists r.B \in \tilde{\T} \} 
				\\ & \cup \{ X  \sqcap \exists r.(X \sqcap B) \sqsubseteq A \mid \exists r.B \sqsubseteq A \in \tilde{\T} \}
				\\ & \cup \{ X  \sqcap A_1 \sqcap A_2 \sqsubseteq A \mid  A_1 \sqcap A_2 \sqsubseteq A \in \tilde{\T} \}
				\\ & \cup \{ X \sqcap \bar{X} \sqsubseteq L \}\\
				 \A' &=  \tilde{\A}
				\\ & \cup \{ \bar{A}(a)\mid  a \in \inds(\tilde{\A}), A\in\Sigma \text{ and } A(a) \not\in \tilde{\A}\} 
				\\ & \cup \{ A(x) \mid A(x) \in \phi \}
				\\ & \cup \{ r(x,y)\mid r(x,y) \in \phi \}
				\\ & \cup \{ X(a) \mid a \in \inds(\tilde{\A}) \} 
				\\ & \cup \{ u(x,a) \mid a \in \inds(\tilde{\A}) \text{ and } x \in var(q)\} 
				\\ & \cup \{ M(x) \mid x \in var(q)\}
				\\ & \cup \{ M(t) \}
				\\ & \cup \{ \bar{X}(a)\mid a \in var(q) \cup \{ t \} \} 
				\\ & \cup \{ u(x,a) \mid a \in var(q) \cup \{ t \} \text{ and } x \in var(q)\}
\end{align*}
Let $q' = q \wedge \bigwedge_{x \in var(q)}  X(x)$. By the proof of \cite[Theorem 2]{lutz2023querying}, $(\tilde{\K},\Sigma) \models q$ iff $Circ_{CP}(\K')\models q'$, where $Circ_{CP}(\K')$ is the \emph{circumscribed KB} with only one minimized concept name: $M$. In this context, the preference relation $<_{CP}$ over interpretations is defined by $\J <_{CP} \I$ when $\Delta^\J = \Delta^\I$ and $M^\J \subsetneq M^\I$. A model of the circumscribed KB $Circ_{CP}(\K')$ is a model $\I$ of $\K'$ such that there are no $\J <_{CP} \I$ and $Circ_{CP}(\K')\models q'$ if every model of $Circ_{CP}(\K')$ satisfies $q'$.

Let us now define an $\ELbot$ WKB $\WKB $ such that $(\tilde{\K},\Sigma) \models q$ iff $\wkb \sat{c}{opt} q'$.
\begin{align*}
\T &= \T' \cup \{ M \sqsubseteq \bot \} \\
\A &= \A' \\
\omega(\tau) &= \infty \text{ for } \tau \in \T' \\
\omega(\alpha) &= \infty \text{ for } \alpha \in \A' \\
\omega(M \sqsubseteq \bot) &= 1
\end{align*}
Let $A_{\tilde{\A}} =\{a\mid A(a)\in\tilde{\A}\}$ for every $A\in\Sigma$.
First, note that by construction of $\K'$ and definition of $\wkb$, it holds that $\optc{\wkb} \geq |\bigcup_{A\in \Sigma} A_{\tilde{\A}}| + |var(q)| +  1$. \smallskip

\noindent($\Leftarrow$) Assume that $(\tilde{\K},\Sigma) \not\models q$ and let $\I$ be a model of $(\tilde{\K},\Sigma) $ such that $\I \not \models q$. 
We assume w.l.o.g.\ that $A^\I=\emptyset$ for concept names $A \not \in \signatureof[\tilde{\K}]$, and likewise for role names not present in  $\tilde{\K}$. 
Let $\J$ be the following interpretation, defined as in the proof of \cite[Theorem 2]{lutz2023querying}, where we assume w.l.o.g.\ that $\Delta^\I \cap (var(q) \cup \{ t \}) = \emptyset$:
\begin{align*}
\Delta^\J &= \Delta^\I \cup var(q) \cup \{ t \} \\
A^\J &= A^\I \cup \{ x \mid A(x) \in \A' \} \text{ for } A \in\NC\setminus \{ M, X \} \\
r^\J &= r^\I \cup \{ (x,y) \mid r(x,y) \in \A'\}  \text{ for } r\in\NR\\
X^\J &= \Delta^\I \\
M^\J &= \bigcup_{A \in \Sigma} A^\J \cup var(q) \cup \{ t \}
\end{align*}

By the proof of \cite[Theorem 2]{lutz2023querying}, $\J \models Circ_{CP}(\K')$ and $\J \not \models q'$. 
It follows that $\J\models\K'=\K_\infty$.  
Moreover, since  $\I$ is a model of $(\tilde{\K},\Sigma) $, for every $A\in\Sigma$, $A^\I=A_{\tilde{\A}}$. Thus, for every $A\in\Sigma$, by construction of $\A'$ and $\J$, $A^\J = A_{\tilde{\A}}\cup \{ x \mid A(x) \in \phi\}\subseteq A_{\tilde{\A}} \cup var(q)$. 
It follows that $\cost{\wkb}{\J} = |M^\J| =|\bigcup_{A\in \Sigma} A_{\tilde{\A}}| + |var(q)| +  1 = \optc{\wkb}$. 
Thus $\J$ is an interpretation with optimal cost such that $\J \not\models q'$. Hence $\wkb \not\sat{c}{opt} q'$. \smallskip

\noindent($\Rightarrow$)
Assume that $\wkb \not\sat{c}{opt} q'$. Let $\I'$ be an interpretation with optimal cost such that $\I' \not\models q'$. It is clear that $\I' \models Circ_{CP}(\K')$ (otherwise a model of $\K'$ such that $\J <_{CP} \I'$ would have a smaller cost than $\I'$). 

As in the proof of \cite[Theorem 2]{lutz2023querying} we claim that there is an interpretation $\I$ with optimal cost w.r.t $\wkb$ s.t. $\I \not\models q'$ and $\I$ satisfies the following properties:
	\begin{enumerate}[label=(\alph*)]
		\item $a \in A^\I $ implies $A(a) \in \A'$ for all $a \in var(q)\cup \{ t\}$,
		\item $(d,e) \in r^\I$ with $r\neq u$ implies $r(d,e) \in \A'$ for all $d,e\in \Delta^\I$ with $\{d,e\} \cap (var(q) \cup\{t\}) \neq \emptyset$, and
		\item $(d,e)\in u^\I$ implies $u(d,e)\in \A'$ for all $d,e \in \Delta^\I$.
	\end{enumerate}

The proof of the claim is the same as that of Claim 2 in the proof of \cite[Theorem 2]{lutz2023querying} as the interpretation they build has optimal cost w.r.t $\wkb$ (we just have to build $\I$ starting with $\I' \models Circ_{CP}(\K')$ previously chosen with optimal cost and since the proof shows that $\I$ is a model of $ Circ_{CP}(\K')$, hence of $\K_\infty$ and that $\I \leq_{CP} \I'$, $M^\I\subseteq M^{\I'}$ and $\I$ has optimal cost). 

Still following the proof of \cite[Theorem 2]{lutz2023querying}, 
we define $\J$ as the restriction of $\I$ to domain $X^\I$ and obtain that $\J \models (\tilde{\K},\Sigma)$ and $\J \not\models q$ (with the same proof as in \cite[Theorem 2]{lutz2023querying}). Hence $(\tilde{\K},\Sigma) \not\models q$.
It follows that \dec{C}{opt}{c} is \twoexptime-hard \wrt combined complexity. 
\bigskip

We now show how to adapt the preceding reduction to show that \dec{C}{b}{c} is \twoexptime-hard. The difficulty is that we need to reduce  to a single instance of \dec{C}{b}{c}, with a single input cost bound $k$. However, from the existing reduction, we only know that $|\bigcup_{A\in \Sigma} A_{\tilde{\A}}| + |var(q)| +  1 \leq \optc{\wkb} \leq   |\inds(\A')|= |\inds(\tilde{\A})| + |var(q)| + 1$, and we cannot easily determine which of these values is the optimal one. The general idea underlying the adapted reduction will be to use a sufficiently high cost bound $k$ that is guaranteed to be equal or greater than the optimal value, while at the same time modifying the KB to ensure that all interpretations with cost at most $k$ (including those with cost higher than the optimal cost) satisfy some properties that ensure that the central arguments of the reduction still go through.

\def\chooserole{pick}

To make this precise, consider the WKB $\K^*_{\omega^*}=(\langle \T^*, \A^*\rangle,\omega^*)$ obtained from $\WKB $ as follows: 
\begin{align*}
\T^* = \,& \T \cup \{  X \sqsubseteq \exists \chooserole. D\} \\
& \cup \{\exists \chooserole. D_1 \sqsubseteq N \}\\
& \cup \{\exists \chooserole. D_2 \sqsubseteq  \bar{A} \mid A \in \Sigma \}\\
& \cup \{D \sqcap \bar{D} \sqsubseteq \bot, X \sqcap \bar{D} \sqsubseteq \bot \} \\
& \cup \{\exists r. L \sqsubseteq L \mid r \in \signatureof[\tilde{\T}]\}\\
& \cup \{N \sqsubseteq \bot, D \sqsubseteq \bot\}\\
\A^* = \,& \A \cup \{N(a) \mid a \in \inds(\tilde{\A})\}\\
& \cup \{D(d_1), D(d_2), D_1(d_1), D_2(d_2)\}\\
& \cup \{\bar{A}(d_1), \bar{A}(d_2) \mid A \in \Sigma\}\\
& \cup \{\bar{X}(d_1), \bar{X}(d_2)\}\\
& \cup \{ u(x,d_1), u(x,d_2) \mid  x \in var(q)\} \\
& \cup \{\bar{D}(x) \mid x \in var(q) \cup \{t\}\} 
\end{align*}
where $d_1,d_2$ are fresh individuals, $D, D_1, D_2$, $\bar{D}$ and $N$ are fresh concept names, and $\chooserole$ is a fresh role name. 
The weight function is defined as follows:
\begin{align*}
\omega^*(M \sqsubseteq \bot) &= 1 \\
\omega^*(D \sqsubseteq \bot) &= |\inds(\tilde{\A})| + |var(q)| + 1 \\
\omega^*(N \sqsubseteq \bot) &= |\inds(\tilde{\A})| + |var(q)| + 1 
\end{align*}
with all other axioms and all ABox assertions assigned infinite weight. 
We keep $q'$ as previously defined and set 
$$k= (|\inds(\tilde{\A})| + |var(q)| + 1) \cdot (|\inds(\tilde{\A})| + 3)$$
To establish \twoexptime-hardness, we prove in what follows that $(\tilde{\K},\Sigma) \models q$ iff $\K^*_{\omega^*} \sat{c}{k} q'$.\smallskip

\noindent($\Leftarrow$) Assume that $(\tilde{\K},\Sigma) \not\models q$ and let $\I$ be a model of $(\tilde{\K},\Sigma) $ such that $\I \not \models q$. 
We may assume w.l.o.g.\ that $ \Delta^\I$ does not contain any element from $var(q) \cup \{ t,d_1, d_2\}$ and that 
all symbols not present in $\tilde{\K}$ are interpreted as $\emptyset$ in $\I$.
We define an interpretation $\J$ as follows:
\begin{align*}
\Delta^\J =& \Delta^\I \cup var(q) \cup \{ t,d_1, d_2\}\\
A^\J =& A^\I \cup \{ x \mid A(x) \in \A^* \} \\ & \text{ for } A \in\NC\setminus (\{ M, X \} \cup \{\bar{A}\mid A \in \Sigma\}) \\ 
r^\J =& r^\I \cup \{ (x,y) \mid r(x,y) \in \A^*\}  \text{ for } r\in\NR \setminus \{\chooserole\}\\
X^\J =& \Delta^\I \\
M^\J =& \bigcup_{A \in \Sigma} A^\J \cup var(q) \cup \{ t \}\\
\bar{A}^\J =& \{ a \mid \bar{A}(a) \in \A^* \} \cup (\Delta^\I \setminus \inds(\tilde{\A})) \text{ for } A \in \Sigma\\
\chooserole^\J =& \{(a,d_1) \mid a \in \inds(\tilde{\A})\} \cup \\&\{(e, d_2) \mid e \in \Delta^\I \setminus \inds(\tilde{\A})\}
\end{align*}
In order to obtain $\K^*_{\omega^*} \not\sat{c}{k} q'$, we shall show that the interpretation $\J$ has cost at most~$k$ and is such that $\J \not \models q'$.

We first show that $\J$ has cost at most $k$. By construction, $\J$ satisfies all of the ABox assertions from $\A^*$. 
Moreover, it can be easily verified that $\J$ satisfies all infinite-weight axioms in $\T^* \cap \T$ due to the similarly with the previous definition of $\J$. 
The only axioms from $\T$ which need to be rechecked are the axioms of the form $A \sqcap \bar{A} \sqsubseteq L$, due to the modified definition of $\bar{A}^\J$. 
We consider these axioms as well as the axioms occurring only in $\T^*$: 
\begin{itemize}
\item Consider the axiom $A \sqcap \bar{A} \sqsubseteq L$ and suppose for a contradiction that $e \in (A \sqcap \bar{A})^\J$. 
Since $e \in A^\J$ and $A \in \Sigma$, we must have either $e \in A^\I$,
 in which case $A(e) \in \tilde{\A}$ and $e \not \in \bar{A}^\J$, 
or $e \in var(q)$, in which case $e \not \in \bar{A}^\J$. In both cases, $e \not \in (A \sqcap \bar{A})^\J$, a contradiction. 
We have thus shown that $(A \sqcap \bar{A})^\J = \emptyset$
so $A \sqcap \bar{A} \sqsubseteq L$ is trivially satisfied.
\item 
Consider $X \sqsubseteq \exists \chooserole. D \in \T^*$ and suppose $e \in X^\J$. 
By definition, we must either have $e \in \inds(\tilde{\A})$ or $e \in \Delta^\I \setminus \inds(\tilde{\A})$. 
In the former case, we have $(e,d_1) \in \chooserole^\J$ and $d_1 \in D^\J$, so $e \in (\exists \chooserole. D)^\J$. 
In the latter case, we have $(e,d_2) \in \chooserole^\J$ and $d_2 \in D^\J$, so again $e \in (\exists \chooserole. D)^\J$. 
\item Consider $D \sqcap \bar{D} \sqsubseteq \bot$ and $X \sqcap \bar{D} \sqsubseteq \bot$. 
Then it is clear from the definition of $\A^*$ and $\J$ that 
$(D \sqcap \bar{D})^\J=\emptyset$ and $(X \sqcap \bar{D})^\J=\emptyset$  as required. 
\item Consider the axiom $\exists \chooserole. D_1 \sqsubseteq N$ and suppose $e \in (\exists \chooserole. D_1)^\J$. 
By construction, this means that $e \in \inds(\tilde{\A})$, hence $e \in N^\J$. 
\item Consider an axiom $\exists \chooserole. D_2 \sqsubseteq \bar{A}$, with $A \in \Sigma$, and $e \in (\exists \chooserole. D_2)^\J$. 
Then due to the definition of $\chooserole^\J$ and $D_2^\J$, we must have $e \in \Delta^\I \setminus \inds(\tilde{\A})$, which immediately gives $e \in \bar{A}^\J$, as required. 
\item Consider an axiom $\exists r. L \sqsubseteq L \in \T^*$. Since $L^\I=\emptyset$ and no $L$ assertions occur in $\A^*$, we
have $L^\J=\emptyset$. It follows that $(\exists r. L \sqsubseteq L)^\J = \emptyset$, so the axiom is trivially satisfied.
\end{itemize}
It follows that no infinite-weight statements are violated. To compute the cost, it 
remains to consider the three axioms having finite weight:
\begin{itemize}
\item Consider the axiom $M \sqsubseteq \bot$ of weight $1$. As $|M^\J| = |\bigcup_{A\in \Sigma} A_{\tilde{\A}}| + |var(q)| +  1$, 
the violations of $M \sqsubseteq \bot$ contribute a total of $|\bigcup_{A\in \Sigma} A_{\tilde{\A}}| + |var(q)| +  1$ to the cost of $\J$. 
\item Consider the axiom $N \sqsubseteq \bot$ of weight $|\inds(\tilde{\A})| + |var(q)| + 1$. As $|N^\J| = |\inds(\tilde{\A})|$, 
the violations of $N \sqsubseteq \bot$ contribute $ |\inds(\tilde{\A})| \cdot (|\inds(\tilde{\A})| + |var(q)| + 1)$. 
\item Consider the axiom $D \sqsubseteq \bot$ of weight $|\inds(\tilde{\A})| + |var(q)| + 1$. As $|D^\J|=2$, 
the violations of this axiom contribute $2 \cdot (|\inds(\tilde{\A})| + |var(q)| + 1)$. 
\end{itemize}
From the preceding items, we arrive at the following cost:
\begin{align*}&(|\bigcup_{A\in \Sigma} A_{\tilde{\A}}| + |var(q)| +  1) \ + \\&(|\inds(\tilde{\A})| + 2) \cdot (|\inds(\tilde{\A})| + |var(q)| + 1)\end{align*} 
Since $|\bigcup_{A\in \Sigma} A_{\tilde{\A}}| \leq |\inds(\tilde{\A})|$, we obtain $\cost{\K^*_{\omega^*}}{\J} \leq k$. 

We now show that $\J \not \models q'$. The argument is similar to one used for the optimal-cost case. 
Indeed, since $q'$ contains the atom $X(v)$ for every $v \in vars(q)$, and $X^\J = \Delta^\I$, 
any match for $q'$ in $\J$ would yield a match for $q$ in $\I$. As $\I$ contains no such match, 
we can infer that $\J \not \models q'$. \medskip

\noindent($\Rightarrow$)
Assume that $\K^*_{\omega^*} \not\sat{c}{k} q'$, and let $\I'$ be an interpretation with optimal cost such that $\I' \not\models q'$. 
We start by establishing the following claim, which shows how we can modify $\I'$ to obtain another interpretation $\I$ with $\cost{\K^*_{\omega^*}}{\I} \leq k$ and $\I \not \models q'$ 
and which satisfies a set of structural properties. \smallskip

\noindent\textbf{Claim.} There exists an interpretation $\I$ with $\cost{\K^*_{\omega^*}}{\I} \leq k$ and $\I \not \models q'$ satisfying the following properties: 
	\begin{enumerate}[label=(\alph*)]
		\item $a \in A^\I $ implies $A(a) \in \A^*$ for all $a \in var(q)\cup \{ t, d_1, d_2\}$, 
		\item $(e,e') \in r^\I$ implies $r(e,e') \in \A^*$ for all $r \not \in \{u, \chooserole\}$ and $e,e'\in \Delta^\I$ with $\{e,e'\} \cap (var(q) \cup\{t,d_1,d_2\}) \neq \emptyset$, 
		\item  $(e,e')\in u^\I$
		implies $u(e,e')\in \A^*$,
		\item $(e,e') \in \chooserole^\I$ implies that $e \in X^\I$, 		
		\item for every $e \in \Delta^\I \setminus \inds(\A^*)$, there exist sequences $e_0, \ldots, e_\ell \in \Delta^\I$ and $r_1, \ldots, r_\ell \in \NR \cap \signatureof[\tilde{\T}]$
		such that $e_0 \in \inds(\tilde{\A})$, $e_\ell=e$, and for every $1 \leq i \leq \ell$,  we have $(e_{i-1}, e_i) \in r_i^\I$ and $e_i \in X^\I$. 
	\end{enumerate}
We thus conserve properties (b) and (c) from the optimal-cost case, and we extend property (a) to cover also the new individuals $d_1,d_2$. 
Property (d) ensures that the new role $\chooserole$ always originates from an element in $X$. The final property (e) enforces that every non-ABox element
is reachable from some original ABox individual  ($\inds(\tilde{\A})$) via a chain of elements in $X$. When combined with the axioms $\exists r. L \sqsubseteq L$, 
this will ensure that any `errors' will be propagated back to some original ABox individual, which in turn will make the query $q'$ hold. \smallskip

\noindent\emph{Proof of claim.} Define $XConn$ as the smallest set that contains $\inds(\tilde{\A})$ and such that 
$e' \in XConn$ whenever $e' \in X^{\I'}$ and there exists $e \in XConn$ and $r \in \NR \cap \signatureof[\tilde{\T}]$ such that  $(e,e') \in r^{\I'}$. 
We define the interpretation $\I$ as follows: 
\begin{align*}
\Delta^\I =& \Delta^{\I'} \cap (\inds(\A^*) \cup XConn)\\
A^\I =& (A^{\I'} \cap \Delta^\I )  \setminus \{ x \in var(q) \\&\cup \{t,d_1,d_2\} \mid A(x) \not \in \A^* \} \\ 
r^\I =& (r^{\I'} \cap (\Delta^\I \times \Delta^\I )) \setminus \{(e,e') \mid r(e,e') \not \in \A^* \text{ and } \\
& \qquad \{e,e'\} \cap (var(q) \cup \{t,d_1,d_2\})\neq\emptyset\} \\
u^\I=& \{(a,b) \mid u(a,b) \in \A^*\}\\
\chooserole^\I =& \chooserole^{\I'} \cap (XConn \times \Delta^\I)
\end{align*}
where $r$ above is neither $u$ nor $\chooserole$. By construction, $\I$ satisfies properties (a)-(e). 
Moreover, as $\I$ is obtained from $\I'$ by removing domain elements 
and reducing the interpretations of some of the concept and role names, we immediately obtain $\I \not \models q'$
from $\I' \not \models q'$. 

To show that $\cost{\K^*_{\omega^*}}{\I} \leq k$, we first show that none of the infinite-weight statements 
is violated: 
\begin{itemize}
\item ABox assertions: we only remove an ABox individual from a concept (or a pair of elements involving an ABox individual from a role)
if no ABox assertion is present to enforce the presence of the individual or pair in the concept / role. Hence since all ABox assertions
were satisfied in $\I'$, they remain satisfied in $\I$. 
\item Axiom $A \sqsubseteq M$: as the axiom is satisfied in $\I'$, the only elements worth checking are the individuals in $var(q) \cup \{t,d_1,d_2\}$.
For those in $var(q) \cup \{t\}$, it suffices to note they belong to $M^\I$, whereas for $d_1$ and $d_2$, we can observe that $d_1 \not \in A^\I$ and 
$d_2 \not \in A^\I$ for all $A \in \Sigma$. 
\item Axiom $\exists u.L \sqsubseteq X$: we show that the axiom is satisfied by showing that $(\exists u. L)^\I = \emptyset$. 
Suppose for a contradiction that $e \in (\exists u. L)^\I$, and take some $e'$ with $(e,e') \in u^\I$ and $e' \in L^\I$. 
As the role $u$ is interpreted in $\I$ exactly as following the ABox $\A^*$, we know that $e \in var(q)$ and moreover that 
$(e'',e') \in u^\I$ for every $e'' \in var(q)$. As $u^\I \subseteq u^{\I'}$ and $L^\I \subseteq L^{\I'}$, and $\I'$ satisfies $\exists u.L \sqsubseteq X$,
we have $e'' \in X^{\I'}$ for every $e'' \in var(q)$. When combined with the satisfaction of $\A^*$ (which contains a `copy' of $q$ using the $var(q)$ individuals),
this allows us to obtain $\I' \models q'$, a contradiction. 
\item Axiom $A \sqcap \bar{A} \sqsubseteq L$ with $A \in \Sigma$: here again, it is enough to check individuals $a \in var(q) \cup \{t,d_1,d_2\}$. 
For each such $a$, we have $\bar{A}(a) \not \in \A^*$, hence $a \not \in \bar{A}^\I$. It follows that $a \not \in (A \sqcap \bar{A})^{\I}$, so the axiom is trivially satisfied.
\item Axiom $X \sqcap \bar{X} \sqsubseteq L$: again, we can focus on individuals $a \in var(q) \cup \{t,d_1,d_2\}$. 
Due to the ABox, we know that $a \in \bar{X} ^\I$. If $a \in X^\I$, then $a \in X^{\I'}$, which would imply $a \in L^{\I'}$ (since $\I'$ must satisfy $X \sqcap \bar{X} \sqsubseteq L$),
which would once again lead to a contradiction with the assumption that $\I' \not \models q'$. 
\item Axioms of the forms $X \sqsubseteq A$, $X  \sqcap A \sqsubseteq B$, $X  \sqcap \exists r.(X \sqcap B) \sqsubseteq A$, and $X  \sqcap A_1 \sqcap A_2 \sqsubseteq A$: a similar argument applies here, as the only relevant elements to check are $a \in var(q) \cup \{t,d_1,d_2\}$, and $a \in X^\I$ would again allow us to obtain $a \in L^{\I'}$, and then contradict $\I' \not \models q'$. 
\item Axiom of the form $X  \sqcap A \sqsubseteq \exists r.(X \sqcap B)$: here we require a different argument, since the restriction of the domain to $XConn$ could potentially remove a witness for $\exists r.(X \sqcap B)$. So consider some $e \in (X  \sqcap A)^\I$. If $e \in var(q) \cup \{t,d_1,d_2\}$,  we use $e \in X^\I$ to obtain a contradiction, as in the preceding items. Otherwise, we have $e \in XConn$. Since $e \in (X  \sqcap A)^\I$, and $e \not \in var(q) \cup \{t,d_1,d_2\}$, it follows that $e \in (X  \sqcap A)^{\I'}$
and thus there exists $e'$ such that $(e,e') \in r^{\I'}$ and $e' \in (X \sqcap B)^{\I'}$. But then we must also have $e' \in XConn$, hence $e' \in \Delta^\I$. 
It follows that $(e,e') \in r^{\I}$ and $e' \in (X \sqcap B)^{\I}$, so $e \in (\exists r.(X \sqcap B))^\I$, as required. 
\item Axiom $D \sqcap \bar{D} \sqsubseteq \bot$: if $e \in (D \sqcap \bar{D})^\I$, then $e \in (D \sqcap \bar{D})^{\I'}$, which is impossible since $\I'$ satisfies the axiom. 
\item Axiom $X \sqcap \bar{D} \sqsubseteq \bot$: analogous to previous item. 
\item Axiom $X \sqsubseteq \exists \chooserole. D$: take some $e \in X^\I$. If $e \in var(q) \cup \{t,d_1,d_2\}$, we obtain a contradiction as in earlier items. Otherwise, 
we have $e \in XConn$. Since $e \in X^{\I'}$ and $\I'$ satisfies the axiom, there must exist $e'$ such that $(e,e')\in \chooserole^{\I'}$ and $e' \in D^{\I'}$. 
If we have $e' \in \{d_1,d_2\}$, then we get $e' \in D^{\I}$ due to the ABox assertions $D(d_1), D(d_2)$. 
Next observe that $e' \in var(q) \cup \{t\}$ is not possible, since by the previous item, $(D \sqcap \bar{D})^\I = \emptyset$, while $x \in \bar{D}^\I$ for each $x \in var(q) \cup \{t\}$. 
For $e' \not \in var(q) \cup \{t, d_1, d_2\}$, then $e' \in D^{\I'}$ implies $e' \in D^{\I}$. Finally, we note that since $e \in XConn$, 
$(e,e')\in \chooserole^{\I'}$ implies $(e,e')\in \chooserole^{\I}$. We therefore obtain $e \in (\exists \chooserole. D)^\I$. 
\item Axiom $\exists \chooserole. D_1 \sqsubseteq N$: Suppose $e \in (\exists \chooserole. D_1)^\I$, and let $e'$ be such that $(e,e')\in \chooserole^{\I}$ and $e' \in D_1^{\I}$.
Due to the definition of $\chooserole^\I$, we have $e \in XConn$ and $(e,e')\in \chooserole^{\I'}$. 
Since $D_1^\I\subseteq D_1^{\I'}$, we also have $e' \in D_1^{\I'}$.
This means that $e \in (\exists \chooserole. D_1)^{\I'}$, hence $e \in N^{\I'}$ due to the satisfaction of the axiom $\exists \chooserole. D_1 \sqsubseteq N$ in $\I'$.
Finally, since $e \in XConn$, we have $e \in X^\I$, so by repeating earlier arguments, we can derive $e \not \in var(q) \cup \{t,d_1,d_2\}$. 
This allows us to infer the desired $e \in N^{\I}$ from $e \in N^{\I'}$.
\item Axiom $\exists \chooserole. D_2 \sqsubseteq  \bar{A}$: Suppose $e \in (\exists \chooserole. D_2)^\I$, and let $e'$ be such that $(e,e')\in \chooserole^{\I}$ and $e' \in D_2^{\I}$.
Reasoning analogously to the previous item, we can show that $e \in (\exists \chooserole. D_2)^{\I'}$, hence $e \in \bar{A}^{\I'}$. Here again we can show that $e \in X^\I$ and thus $e \not \in var(q) \cup \{t,d_1,d_2\}$. It then follows from $e \in \bar{A}^{\I'}$ that $e \in \bar{A}^{\I}$. 
\item Axiom $\exists r. L \sqsubseteq L$ for $r \in \signatureof[\tilde{\T}]$: Suppose that $e \in (\exists r. L)^\I$, and take some $e'$ such that $(e,e') \in r^\I$ and $e' \in L^\I$. 
Since $r^\I\subseteq r^{\I'}$ and $L^\I\subseteq L^{\I'}$, we have $(e,e') \in r^{\I'}$ and $e' \in L^{\I'}$, hence $e \in (\exists r. L)^{\I'}$. As $\I'$ satisfies $\exists r. L \sqsubseteq L$,
this yields $e \in L^{\I'}$. We can then use $e \in L^{\I'}$ and $\I' \not \models q'$ to infer that $e \not \in var(q) \cup \{t,d_1,d_2\}$. This in turn means that $e \in L^\I$ must hold as $e \in L^{\I'}$.
\end{itemize}
We next consider the three axioms with finite weights: 
$$M \sqsubseteq \bot \quad  N \sqsubseteq \bot \quad D \sqsubseteq \bot$$
Here the argument is easy, since the definition of $\I$ ensures $M^\I \subseteq M^{\I'}$, $N^\I \subseteq N^{\I'}$,
and $D^\I \subseteq D^{\I'}$. From this, we can immediately infer $\cost{\K^*_{\omega^*}}{\I} \leq \cost{\K^*_{\omega^*}}{\I'}$, 
and since we have assumed that $\cost{\K^*_{\omega^*}}{\I'}  \leq k$, we get $\cost{\K^*_{\omega^*}}{\I} \leq k$, as required. 
 (\emph{End proof of claim})\medskip
 
 Now we take the interpretation $\I$ given by the preceding claim, and we define $\J$ 
 as the restriction of $\I$ to elements in $X^\I$. We aim to show that $\J \models (\tilde{\K},\Sigma)$ and $\J \not\models q$. 
 The latter is immediate since $\J$ is a restriction of $\I$ and $\I \not \models q'$, with $q'$ containing $q$ as a subquery. 
 We thus concentrate on showing $\J \models (\tilde{\K},\Sigma)$, starting by considering the different forms of axioms in $\tilde{\T}$: 
 \begin{itemize}
 \item Axiom $\top \sqsubseteq A \in \tilde{\T}$: In this case, $\T^*$ contains $X \sqsubseteq A$. Since $\Delta^\J=X^\I$
 and $X^\I \subseteq A^\I$, we get $\Delta^\J \subseteq A^\I$.
 \item Axiom $A \sqsubseteq B \in \tilde{\T}$: In this case, $\T^*$ contains $X  \sqcap A \sqsubseteq B$. Again, using 
 $\Delta^\J=X^\I$ and the fact that $\I$ is a model of $\T^*$, we have $A^\J = A^\I \cap X^\I \subseteq B^\I \cap X^\I = B^\J$. 
 \item Axiom $ A \sqsubseteq \exists r.B \in \tilde{\T}$: In this case, $\T^*$ contains $X  \sqcap A \sqsubseteq \exists r.(X \sqcap B)$. 
 We have $A^\J = A^\I \cap X^\I \subseteq X^\I \cap (\exists r.(X \sqcap B))^\I \subseteq (\exists r. B)^\J$. 
 \item Axiom $\exists r.B \sqsubseteq A \in \tilde{\T}$: In this case, $X  \sqcap \exists r.(X \sqcap B) \sqsubseteq A$ appears in $\T^*$. 
 We then have $(\exists r. B)^\J = (X^\I \cap (\exists r. (B \sqcap X))^\I) \subseteq X^\I \cap A^\I \subseteq A^\J$. 
 \item Axiom $A_1 \sqcap A_2 \sqsubseteq A \in \tilde{\T}$: In this case, $\T^*$ contains $X  \sqcap A_1 \sqcap A_2 \sqsubseteq A$. 
 We have $(A_1  \sqcap A_2)^\J = X^\I \cap A_1^\I \cap A_2^\I \subseteq X^\I \cap A^\I = A^\J$. 
 \end{itemize}
As for the ABox assertions in $\tilde{A}$, it suffices to observe that they are satisfied in $\I$ and that $a \in X^\I$ 
for every $a \in \inds(\tilde{\A})$. 

It remains to show that $\J$ respects $\Sigma$. To this end, it is useful to observe 
that due solely to the ABox assertions in $\A^*$ involving the concept names 
$M$, $N$, and $D$, the cost of $\I$ must be at least:
\begin{align*}
&(|\bigcup_{A\in \Sigma} A_{\tilde{\A}}| + |var(q)| +  1)\ + \\&(|\inds(\tilde{\A})| + 2) \cdot (|\inds(\tilde{\A})| + |var(q)| + 1)
\end{align*}
and it is at most 
$$k= (|\inds(\tilde{\A})| + |var(q)| + 1) \cdot (|\inds(\tilde{\A})| + 3)$$
by our earlier assumption.  
We now consider the two possible ways that the closed predicates may not be respected: 
\begin{itemize}
\item Suppose that $a \in A^\J$, $a \in \inds(\tilde{\A})$, but $A(a) \not \in \tilde{\A}$. Then $a \in A^\I$ and 
$a \in \bar{A}^\I$ (due to $\bar{A}(a) \in \A^*$). It can then be shown, using the axioms $A \sqcap \bar{A} \sqsubseteq L$ and $\exists u. L \sqsubseteq X$ and assertions $u(x,a)$ for $x \in var(q)$, that we must have $\I  \models q'$, which is a contradiction. 
\item Suppose instead that $e \in A^\J$ but $e \not \in  \inds(\tilde{\A})$. Since $e \in \Delta^\J$, it follows that $e \in X^\I$, and hence that $e \not \in var(q) \cup \{t,d_1,d_2\}$. Moreover,  
due to the axiom $X \sqsubseteq \exists \chooserole. D$, there must exist $e' \in D^\I$ such that $(e,e') \in \chooserole^\I$. 
There are three possibilities for $e'$ to consider: 
\begin{itemize}
\item $e' \not \in \{d_1, d_2\}$, in which case $|D^\I|\geq 3$. The addition of a further $D$ element, together with $D \sqsubseteq \bot$ of weight $|\inds(\tilde{\A})| + |var(q)| + 1$, introduces an extra $|\inds(\tilde{\A})| + |var(q)| + 1$, taking the total cost of $\I$ above $k$, a contradiction.
\item $e'=d_1$, in which case we must have $e \in N^\I$. But since we assumed that $e \not \in  \inds(\tilde{\A})$, this means that we have an additional element in $N$, and thus must pay a further penalty of $|\inds(\tilde{\A})| + |var(q)| + 1$ for violating $N \sqsubseteq \bot$, again taking the cost of $\I$ above $k$. 
\item $e' = d_2$, in which case we must have $e \in \bar{A}^\I$. But due to the axiom $A \sqcap \bar{A} \sqsubseteq L$, we must have $e \in L^\I$. Furthermore, by point (e) of the claim, there exist sequences $e_0, \ldots, e_\ell \in \Delta^\I$ and $r_1, \ldots, r_\ell \in \NR \cap \signatureof[\tilde{\T}]$
such that $e_0 = a \in \inds(\tilde{\A})$, $e_\ell=e$, and for every $1 \leq i \leq \ell$,  we have $(e_{i-1}, e_i) \in r_i^\I$ and $e_i \in X^\I$. 
By a simple inductive argument, using the axioms $\exists r. L \sqsubseteq L$, we can infer that the individual $e_0 = a \in \inds(\tilde{\A})$ belongs to $L^\I$. 
But this allows us, reusing earlier arguments, to conclude that $\I \models q'$, a contradiction. 
\end{itemize}
\end{itemize}
As we reach a contradiction in each of the cases, we may conclude that  $\J$ respects $\Sigma$.

A final observation is that since the weights and input $k$ value used in the reduction are polynomial, we obtain \twoexptime-hardness even for unary 
encoding. 
\end{proof}


\newcommand{\Idag}{\I^\dagger}

\section{Proofs for Section \ref{sec:data}}

\subsection{Proofs of Propositions \ref{prop:boundedSizekInterpretation} and \ref{prop:CQboundedPossible}}

\PropBoundedSizekInterpretation*
\begin{proof}

Let $\WKB$ be a WKB and \emph{k} an integer. Let $\I$ be an interpretation with $\cost{\wkb}{\I} \leq k$.
Let $\J$ be the $S$-filtration of $\I$ with $S=sub(\T)$, defined as in the proof of Proposition \ref{prop:boundedModelALCO} (Appendix \ref{app-filtration}). 
We know that $|\Delta^\J| \leq |\inds(\K)| + 2^{|\T|}$. We show that $\cost{\wkb}{\J} \leq k$.

To do so, we show that $\cost{\wkb}{\J} \leq \cost{\wkb}{\I}$. First note that by construction of $\J$, $ \vio{\A}{\J} = \vio{\A}{\I}$. 
	
We are now left to show that for every $\tau \in \T$ we have $|\vio{\tau}{\J}| \leq |\vio{\tau}{\I}|$. Let $\tau = C \sqsubseteq D \in \T$. We need to show that $|(C\sqcap\neg D)^\J|\leq |(C\sqcap\neg D)^\I|$. Since $C,D \in S$ this directly follows from Lemma \ref{lem:mainEquiv}. 
\end{proof}

\PropCQboundedPossible*
\begin{proof}
Let $\WKB$ be a WKB, \emph{k} an integer and \emph{q} a conjunctive query. Let $\I$ be an interpretation with $\cost{\wkb}{\I} \leq k$ s.t. $\I \models q$.
Let $\J$ be the $S$-filtration of $\I$ with $S=sub(\T)$ (as defined in Appendix \ref{app-filtration}).  With the proofs of Proposition \ref{prop:boundedModelALCO} and Proposition \ref{prop:boundedSizekInterpretation}, we know that $|\Delta^\J| \leq |\inds(\K)| + 2^{|\T|}$ and $\cost{\wkb}{\J} \leq k$. We want to show that $\J \models q$.

Let $\pi : var(q) \rightarrow \Delta^\I$ be a match for q in $\I$ and define $\tilde{\pi} : var(q) \rightarrow \Delta^\J$ by
\begin{itemize}
\item $\tilde{\pi}(x) = \pi(x)$ if $\pi(x) \in \Delta^\I \cap \inds(\K)$ and
\item$\tilde{\pi}(x) = [\pi(x)]_S$ if $\pi(x) \in \Delta^\I \setminus \inds(\K)$.
\end{itemize}
It follows from Lemma \ref{lem:mainEquiv}, the definition of $S$-filtration, and the fact that $\pi$ is a match for \emph{q} in $\I$ that $\tilde{\pi}$ is a match for $q$ in~$\J$, as required. 
\end{proof}

\subsection{Proof of Proposition \ref{prop:CQboundedCertain}}

Consider an $\mathcal{ALCO}$ WKB $\WKB$, a BCQ~$q$, and a finite bound $k$,
and suppose that there exists an interpretation $\I$ with $\cost{\wkb}{\I} \leq k$ such that $\I \not \models q$. 
We may assume w.l.o.g.\ that $q$ does not mention any individuals not present in $\K$. 

In order to prove Proposition \ref{prop:CQboundedCertain}, we will suitably adapt the constructions 
in \cite{Maniere} to transform $\I$ into a finite interpretation $\J$, of size polynomial in $|\A|$, such that 
$\cost{\wkb}{\J} \leq k$ and $\J \not \models q$. 

\paragraph{Translating input WKB into a standard KB in normal form}
As the constructions in \cite{Maniere} work on standard KBs (without weights) in normal form,
our first step will be to replace the initial WKB $\K$ with an $\ALCO$ KB in normal form. 
Essentially the idea is to remove all inclusions with finite weights and instead introduce new concept
names into the TBox to be able to keep track of axiom violations. We will also prune the initial ABox to keep only those 
assertions that are satisfied in the interpretation $\I$. 

We recall that an $\ALCO$ KB is in normal form if all of its TBox axioms are in one of the following forms:
\begin{align*}
\axiom{\top}{D} \qquad \axiom{A}{D} \qquad \axiom{\{a\}}{A}  \qquad 
\axand
\\
\axexistsleft
\qquad
\axexistsright
\qquad
\axnotleft
\qquad
\axnotright
\end{align*}
where $A, A_1, A_2, B \in \NC$, $D \in \NC \cup \{\{a\} \mid \NI\}$, and $R \in \NR$. 
It is well known that an arbitrary $\ALCO$ TBox can be transformed (possibly through introduction of new concept names)
into an $\ALCO$ TBox in normal form. As the transformation may possibly introduce new concept names, 
the new TBox may not be equivalent to the original one, but it will be a conservative extension, 
which is sufficient for our purposes.  

We now proceed as follows to convert our initial WKB into a KB in normal form:
\begin{enumerate}
\item For every 
 $\tau = C \sqsubseteq D \in \T \setminus \T_\infty$, 
let $A_\tau$ be a fresh concept name (i.e. occurring neither in $\K$ nor $q$), and set $$\T_\tau= \{A_\tau \sqsubseteq C \sqcap \neg D, C \sqcap \neg D \sqsubseteq A_\tau\}$$

\item Apply the normalization procedure to the TBox 
$$\T_\infty^{vio} = \T_\infty \cup \bigcup_{\tau \in \T \setminus \T_\infty} \T_\tau$$ 
and let $\Tnew$ denote the resulting TBox. We can assume w.l.o.g.\ that all concept names in $\signatureof[\Tnew] \setminus \signatureof[\T_\infty^{vio}]$
occur neither in $q$ nor $\K$. 
\end{enumerate} 
We set $\Knew = \langle \Tnew, \Anew \rangle$, where $$\Anew = \{\alpha \in \A \mid \I \models \alpha\} \cup \{A_\top(a) \mid a \in \indsof{\A}\}$$
with $A_\top$ a fresh concept, not in $\Tnew$ nor $q$.
Note that the assertions $A_\top(a)$ simply serve to ensure that no individuals from $\K$
are lost during the translation (which will be convenient for the later constructions and proofs).

The following lemma summarizes the properties of $\Knew$.  Items (ii)-(iv) basically correspond to 
$\Tnew$ being a conservative extension of $\T_\infty^{vio}$, while (v) is a direct consequence 
of (iii) and the definition of $\T_\infty^{vio}$. 

\begin{lemma}\label{newkb}
The KB $\Knew$ satisfies the following:
\begin{enumerate}[label=(\roman*),leftmargin=25pt]
\item $\inds(\K) = \inds(\Knew)$
\item 
$\signatureof[\T] \subseteq \signatureof[\T_\infty^{vio}] \subseteq \signatureof[\Tnew]$;
\item 
every model of $\Tnew$ is a model of $\T_\infty^{vio}$;
\item
for every model $\I_1$ of $\T_\infty^{vio}$, there exists a model $\I_2$ of $\Tnew$ 
with $\Delta^{\I_1}=\Delta^{\I_2} $ such that 
$\cdot^{\I_2}$ and $\cdot^{\I_1}$ coincide on all concept and role names except those in $\signatureof[\Tnew] \setminus \signatureof[\T_\infty^{vio}]$;
\item for every $\tau \in \T \setminus \Tinf$ and every model $\I'$ of $\Knew$,
$\vio{\tau}{\I'} = A_{\tau}^{\I'}$.
\end{enumerate}
\end{lemma}

Using the preceding lemma, we can modify our original interpretation $\I$
to get a model of the KB $\Knew$:

\begin{lemma}\label{norm-model}
There exists a model $\I^\dagger$ of $\Knew$ with $\Delta^{\I^\dagger} = \Delta^\I$
such that $\cdot^{\I^\dagger}$ and $\cdot^{\I}$ coincide on $\signatureof[\K] \cup \signatureof[q]$. 
In particular, this means that:
\begin{itemize}
\item $\vio{\A}{\I^\dagger} = \vio{\A}{\I}$
\item for every  $\tau \in \T \setminus \Tinf$: $\vio{\tau}{\I^\dagger} = \vio{\tau}{\I}$
\item $\I^\dagger \not \models q$
\end{itemize}
\end{lemma}
\begin{proof}
We first create an interpretation $\I^\diamond$ which is the same as $\I$ except for the 
concept names $A_\tau \in  \signature[\T_\infty^{vio}] \setminus \signatureof[\T]$, for which we 
set $(A_\tau)^{\I^\diamond}=(C \sqcap \neg D)^\I$, and the concept name $A_\top$
which we interpret as $\indsof{\A}$ (note that we may assume $\indsof{\A} \subseteq \Delta^\I$ due to the weak SNA). 
This ensures that $\I^\diamond$ satisfies the axioms in $\T_\tau$ and the new assertions in $\Anew$,
 and since we have not modified the interpretation of any other symbols, 
 $\I^\diamond$ will be a model of  both $\Tinf$ and $\Anew$. 
Thus, the interpretation $\I^\diamond$
is a model of $ \langle \T_\infty^{vio}, \Anew \rangle$. 
Hence, by Lemma \ref{newkb} (iv), there exists a model $\I^\dagger$ of $\Tnew$
which can be obtained from $\I^\diamond$ solely by changing the interpretation of concept names in 
$\signatureof[\Tnew] \setminus \signatureof[\T_\infty^{vio}]$. In particular, $\I^\dagger$ is a model of $\Anew$
and hence of the KB $\Knew$.  It follows from the construction of $\I^\dagger$ that
$\Delta^{\I^\dagger} = \Delta^\I$ and that $\cdot^{\I^\dagger}$ and $\cdot^{\I}$ coincide on $\signatureof[\K] \cup \signatureof[q]$.
The first and third items follow immediately (for the third item, recall that $\I \not \models q$), and
the second item holds due to $(A_\tau)^{\I^\dagger}=(A_\tau)^{\I^\diamond}=(C \sqcap \neg D)^\I$. 
\end{proof}

As the model $\I^\dagger$ of $\Knew$ given in Lemma \ref{norm-model} 
satisfies the same key properties as $\I$ (i.e.\ not satisfying $q$ and having cost $\leq k$),
it can thus be used in place of $\I$ in the following constructions.

\paragraph{Interlacing construction to regularize the interpretation}
The second step will be to apply the interlacing construction from \cite{Maniere}
to our interpretation $\I^\dagger$ to obtain a more well-structured interpretation
that retains the essential properties of $\Idag$, i.e.\ it does not entail $q$ and its cost does not exceed $k$. 

The interlacing construction starts with the definition of the existential 
extraction, which is a tree-shaped domain. The definition, which we recall next, 
uses the alphabet $\Omega$ consisting of all $\rolestyle{R.A}$ such that $\existsrole{R.A}$ is the RHS of an axiom in  
the considered TBox ($\Tnew$ in our case),
and for every $\rolestyle{R.A} \in \Omega$, a function $\successor[\Idag]_{\rolestyle{R.A}}$
that maps every element $e \in (\existsrole{R.A})^{\Idag}$ to an element $e' \in \domain{\Idag}$ such that $(e, e') \in \rolestyle{R}^{\Idag}$ and $e' \in \rolestyle{A}^{\Idag}$.

\begin{definition}[Existential extraction]
\label{def:existential-extraction}
Build the following mapping inductively over the set $\indsof{\Knew} \cdot \Omega^*$:
\[
\begin{array}{rcl}
\exextomod :
\indsof{\Knew} \cdot \Omega^* 	& \rightarrow 	& \domain{\Idag} \cup \{ \uparrow \}
\\
a 							& \mapsto 		& a \\
\end{array}
\]
\[
\begin{array}{rcl}
w \cdot \rolestyle{R.A} 			& \mapsto 		& \left\{ \begin{array}{ll}
													\uparrow 								& \textrm{if }  \exextomodof{w} = \; \uparrow \\
								&
																							  \textrm{ or } \exextomodof{w} \notin (\existsrole{R.A})^{\Idag}
													\\
													\successorof[\Idag]{\exextomodof{w}}{R.A} 	& \textrm{otherwise}
												  \end{array} \right.
\end{array}
\]
where $\uparrow$ is a fresh symbol witnessing the absence of a proper image for an element of $\indsof{\Knew} \cdot \Omega^*$.
The \emph{existential extraction} of $\Idag$ is $\exexof{\I} := \{ w \mid w \in \indsof{\Knew} \cdot \Omega^*, \exextomodof{w} \neq \; \uparrow \}$. 
\end{definition}

The preceding definition only differs from the one in \cite{Maniere} in one way: it uses the set $\indsof{\Knew}$ of individuals in the considered KB 
rather than the set of individuals in the considered ABox. This is because we work with $\ALCO$ (and thus individuals may appear in the TBox) whereas the original 
construction was formulated for $\mathcal{ALCHI}$. Fortunately, although nominals were not considered in \cite{Maniere},
they are properly handled by the constructions and do not require any notable modifications.

One part of the construction that we will need to modify to suit our setting is  
the domain of interest $\deltastar$, which intuitively is the part of $\I$ we wish to remain untouched:

\[
\deltastar := \individuals(\Knew) \, \cup \,  \bigcup_{\tau \in \T \setminus \Tinf} \vio{\tau}{\Idag}
\]
Whereas the original definition of $\deltastar$ takes the set of ABox individuals together with all elements occurring 
in a match of the input (counting) CQ, we must consider nominals from the TBox and all elements that occur in a violation of 
an axiom in $\T \setminus \Tinf$. Note however that since in our case we start from an interpretation $\Idag$ that does not 
contain any matches for $q$, our set $\deltastar$ in fact is a superset of the one considered in \cite{Maniere}.

We can now recall the function $f^*$, defined exactly as in \cite{Maniere} but using the modified version of $\deltastar$: 

\[
\begin{array}{cccl}
\interlace : &
\domain{\unfoldingof{\I}} 	& \rightarrow 	& \deltastar \uplus (\domain{\unfoldingof{\I}} \setminus \deltastar)
\\[3pt] &
w 							& \mapsto		& 	\left\{ \begin{array}{ll}
													\exextomodof{w}
														& \textrm{if } \exextomodof{w} \in \deltastar
													\\
													w 
														& \textrm{otherwise}
												\end{array} \right.
\end{array}
\]

With these notions in hand, we can now recall that definition of the interlacing of $\Idag$ w.r.t.\ the function $f^*$:

\begin{definition}[$\interlace$-interlacing]
The \emph{$\interlace$-interlacing} $\interlacingofstar{\I}$ of $\Idag$ is the interpretation whose domain is 
$\domain{\interlacingofstar{\I}} := \interlaceof{\exexof{\I}}$ and which interprets concept and role names as follows:
\begin{align*}
\rolestyle{A}^{\interlacingofstar{\I}} ~ := ~ &
	\{ \interlaceof{u} \mid u \in \exexof{\I}, \exextomodof{u} \in \rolestyle{A}^{\Idag} \}
\\
\rolestyle{R}^{\interlacingofstar{\I}} ~ := ~ &
	\{ (a, b) \mid a, b \in \indsof{\Knew} \wedge \Knew \models \roleassertion{R}{a}{b} \}
	\\ & \cup ~ 
	\{ (\interlaceof{u}, \interlaceof{u \cdot \rolestyle{R.B}}) \mid u, u \cdot \rolestyle{R.B} \in \exexof{\Idag} \}
\end{align*}
\end{definition}

The preceding definition is a simplified version of Definition 20 from \cite{Maniere} as we are not considering role inclusions nor inverse roles. 
We also make a slight modification -- using $\indsof{\Knew}$ rather than only individuals from the ABox  -- to make it compatible with nominals.

We have phrased the definition directly in terms of our desired function $f^*$, but other functions $f'$ with domain $\exexof{\I}$ can be used instead. 
Depending on which $f'$ is used to define the interlacing, 
the resulting interpretation $\I'$ may or may not be a model of the considered KB. 
It was shown in \cite{Maniere} that if $f'$ is pseudo-injective, then the $f'$-interleaving is a model
and moreover maps homomorphically into the starting interpretation. 
This property (stated in Lemma \ref{thm:meta-interlacing-is-model} below, phrased for our KB $\Knew$) was proven for $\mathcal{ALCHI}$ KBs in normal form, 
but is easily shown to also hold for
$\ALCO$ KBs in normal form. 

\begin{definition}
A function $f' : \exexof{\I} \rightarrow E$ is \emph{pseudo-injective}\index{pseudo-injective} if: for all $u, v \in \exexof{\I}$, if $f'(u) = f'(v)$, then $\exextomodof{u} = \exextomodof{v}$. 
\end{definition}

\begin{lemma}
\label{thm:meta-interlacing-is-model}
If $f': \exexof{\I} \rightarrow E$ is pseudo-injective, then the $f'$-interlacing $\interlacingof{\I}$ is a model of $\Knew$ and the following mapping is a homomorphism from $\interlacingof{\I}$ to $\Idag$:
\[
\begin{array}{cccl}
\intertomod : &
\domain{\interlacingof{\I}} 	& \rightarrow 	& \domain{\Idag} 
\\ &
f'(u) 				& \mapsto		& f(u) 
\end{array}
\]
As $f'$ is pseudo-injective, $\intertomod$ is well defined.
\end{lemma}

It was proven in \cite{Maniere} that $f^*$ is pseudo-injective,
and the same arguments apply to
our modified $f^*$. 

\begin{lemma}
$f^*$ is pseudo-injective. 
\end{lemma}

It follows from the preceding lemmas that $\interlacingofstar{\I}$ is a model of $\Knew$ which maps 
homomorphically into $\Idag$ via the mapping~$\sigma$. 
In fact, due to the way $f^*$ is defined, we can be more precise about the homomorphism $\sigma$:

\begin{lemma}\label{id-inds}
The homomorphism $\sigma$ from $\interlacingofstar{\I}$ to $\Idag$
is such that $\sigma(a) = a$ for all $a \in \indsof{\Knew}$. 
\end{lemma}

We are now ready to show that $\interlacingofstar{\I}$ retains the desired properties of $\Idag$ (and $\I$). 

\begin{lemma} 
$\interlacingofstar{\I} \not \models q$.
\end{lemma}
\begin{proof}
We could obtain the result by examining the proof of Lemma 5 in \cite{Maniere}, which shows that 
the query matches in $\interlacingofstar{\I}$ (there denoted $\I'$ and defined w.r.t.\  the original definition of $\Delta^*$) are injectively mapped into $\Idag$
(and hence there cannot be any additional matches in $\interlacingofstar{\I}$).  However, it turns out that
the argument is much simpler in our case, since we are start from an interpretation $\Idag$ with $\Idag \not \models q$.
The existence of the homomorphism $\sigma$ from $\interlacingofstar{\I}$ 
to $\Idag$ that is the identify on $\indsof{\Knew}$ (see Lemmas \ref{thm:meta-interlacing-is-model} and \ref{id-inds}), 
hence also on $\indsof{q}$,  
means that any match of $q$ in $\interlacingofstar{\I}$ can be reproduced in $\Idag$. 
Therefore, from $\Idag \not \models q$, we immediately obtain $\interlacingofstar{\I} \not \models q$.
\end{proof}

\begin{lemma} \label{interleaving-cost}
The cost of $\interlacingofstar{\I}$ does not exceed that of $\Idag$:
$\cost{\wkb}{\interlacingofstar{\I}} \leq \cost{\wkb}{\Idag}$. In particular, we have that
$\vio{\tau}{\interlacingofstar{\I}} \subseteq \vio{\tau}{\Idag} \subseteq \Delta^*$ for every $\tau \in \T \setminus \T_\infty$. 
\end{lemma}
\begin{proof}
First note that since $\interlacingofstar{\I}$ is a model of $\Knew$, it is a model of $\Anew$ and hence satisfies all assertions in 
 $\{\alpha \in \A \mid \I \models \alpha\}= \{\alpha \in \A \mid \Idag \models \alpha\}$.
Since $\Idag \not \models \alpha$ for $\alpha \in \A \setminus \Anew$, it follows that 
$$\vio{\A}{\interlacingofstar{\I}} \subseteq \vio{\A}{\Idag}$$

Let us now consider violations of TBox axioms. First note that since 
$\interlacingofstar{\I}$ is a model of $\Knew$, it is a model of $\T_\infty^{vio}$ (Lemma \ref{newkb}, item (iii)).
It follows that $\interlacingofstar{\I} \models \tau$ for every $\tau \in \T_\infty$.
Now consider some $\tau \in \T \setminus \T_\infty$ and suppose that $e \in \vio{\tau}{\interlacingofstar{\I}}$. 
By Lemma \ref{newkb}, item (v),  $\vio{\tau}{\I^*} = A_{\tau}^{\I^*}$, so $e \in A_{\tau}^{\I^*}$.
From the definition of $\interlacingofstar{\I}$, we know that there exists $w \in \domain{\unfoldingof{\I}}$
such that $e = f^*(w)$. 
As $\sigma$ is a homomorphism from $\interlacingofstar{\I}$ to $\Idag$ (Lemma \ref{thm:meta-interlacing-is-model}) and $e = f^*(w) \in A_{\tau}^{\I^*}$, we must also have $f(w) \in A_{\tau}^{\Idag}$. 
As $\Idag$ is a model of $\Knew$, we can again apply Lemma \ref{newkb} item (v)
to get $\vio{\tau}{\Idag} = A_{\tau}^{\Idag}$. We thus obtain $f(w) \in \vio{\tau}{\Idag}$,
and thus $f(w) \in \deltastar$ (due to our definition of $\deltastar$). 
But $f(w) \in \deltastar$ implies that $e=f^*(w)=f(w)$, and hence that $e \in A_{\tau}^{\Idag}$. 
We have thus shown that 
$\vio{\tau}{\interlacingofstar{\I}} \subseteq \vio{\tau}{\Idag}$. 

As every violation of an axiom or assertion from $\wkb$ in the interpretation $\interlacingofstar{\I}$
is also present in $\Idag$, we can conclude that
$\cost{\wkb}{\interlacingofstar{\I}} \leq \cost{\wkb}{\Idag}$. 
\end{proof}

\paragraph{Quotient construction to get bounded-size interpretation} 
It now remains to `shrink' the  interpretation $\interlacingofstar{\I}$ to obtain an interpretation 
with the same properties but having the required size. 
To do so, we can proceed exactly as in 
 Chapter 3.4 of \cite{Maniere}, by  defining a suitable equivalence relation, 
and then considering the quotient interpretation obtained by merging elements from the same equivalence class.

We will now describe how the equivalence relation $\sim_n$ is defined, but without giving every detail as
the definition is rather involved and we will not be modifying it except by  
using our own definition of $\deltastar$.
The definition of $\sim_n$ involves the notion of the $n$-neighbourhood of an element $d$ in $\I^*$ relative to $\deltastar$,  denoted  $\mathcal{N}^{\I^*,\deltastar}_n(d)$,
consisting of the  elements in $\Delta^{\I^*}$ that can be reached by taking at most $n$ `steps' along role edges, starting at $d$ 
and stopping whenever an element of $\deltastar$ is reached. 
If $d\in \deltastar$, then $d$ itself is the only element in its $n$-neighbourhood (for any $n$). 
However, when $d \in \Delta^{\I^*} \setminus \Delta^*$,
we know that $d = a w$ for some $a \in \indsof{\Knew}$ and $w \in \Omega^*$. 
Using the tree-shaped structure of the domain $\domain{\unfoldingof{\I}}$, 
we can identify a unique `root' prefix $r_{n,d}$ of $d=a w $
such that $f^*(r_{n,d}) \in \mathcal{N}^{\I^*,\deltastar}_n(d)$ and 
for every $d' \in \mathcal{N}^{\I^*,\deltastar}_n(d)$, there is a unique word 
$w^{d'}_{n,d} \in \Omega^*$ such that $d'=f^*(r_{n,d} \cdot w_{n,d}^{d'})$ and $|w^{d'}_{n,d}| \leq 2n$. 
With this uniform way of refer to the elements in neighbourhood of a considered element $d$, 
we can define a function $\chi_{n,d}$ whose output tells us for each word $w \in \Omega^*$ with $|w| \leq 2n$ 
whether there is an element in the neighbourhood whose word is $w$, 
and if so, whether that element belongs to $\deltastar$, and if not, which concept names it satisfies in $\I^*$. 
The equivalence relation $\sim_n$ then groups together those elements $d,e \in \Delta^{\I^*} \setminus \Delta^*$
which have the same associated word (i.e.\ $w^{d}_{n,d} = w^{e}_{n,e}$) and same associated function ($\chi_{n,d}=\chi_{n,e}$), 
plus an additional condition on the length of $d$ and $e$ (namely, $|d| = |e| \textrm{ mod } 2|q|+3$). 

With the equivalence relation $\sim_n$ at hand, it now 
suffices to merge elements which are equivalent w.r.t.\ $\sim_{|q|+1}$.
Formally, we consider the quotient interpretation $\J:= \I^* / \sim_{|q|+1}$
whose domain is $\Delta^\J= \{[e]_{\sim_{|q|+1}} \mid e \in \Delta^{\I^*}\}$ 
and whose interpretation function 
$\cdot^\J$ is as follows:
\begin{itemize}
\item $A^\J = \{[e]_{\sim_{|q|+1}} \mid e \in A^{\I^*}\}$
\item $R^\J = \{([d]_{\sim_{|q|+1}},[e]_{\sim_{|q|+1}}) \mid (d,e) \in R^{\I^*}\}$
\item $a^\J= [a^{\I^*}]_{\sim_{|q|+1}} $
\end{itemize}
where  $[e]_{\sim_{|q|+1}}$ is the equivalence class of $e$ w.r.t.\ 
$\sim_{|q|+1}$.

It has been shown that $\J$ remains a model of the considered KB and does not contain
any additional query matches. Given that the interpretation $\I^*$ we consider here is built in exactly 
the same way as the interleaving $\I'$ considered in \cite{Maniere}, except that we work 
with a larger domain of interest $\deltastar$ (which includes the one used to build $\I'$),
exactly the same arguments can be used to show Lemma \ref{quotient-model} below. 
We point out that the presence of nominals in $\Knew$ is not problematic 
as all of its individuals (whether present in the ABox or TBox) are included in $\deltastar$ and are therefore left
untouched by the construction. 

\begin{lemma}\label{quotient-model}
The interpretation $\J$ is a model of $\Knew$ such that $\J \not \models q$. 
\end{lemma}

It remains to show that the quotient operation does not increase the cost. 

\begin{lemma}
The cost of $\J$ is no higher than the cost of $\I^*$: $\cost{\wkb}{\J} \leq \cost{\wkb}{\interlacingofstar{\I}}$. 
\end{lemma}
\begin{proof}
First note that since $\J$ is a model of $\Knew$, it satisfies all assertions in $\Anew$. 
It follows that $\vio{\A}{\J} \subseteq \vio{\A}{\interlacingofstar{\I}}$. 
The fact that $\J$ is a model of $\Tnew$ means that it satisfies every $\tau \in \Tinf$. 
Consider now an axiom $\tau \in \T \setminus \Tinf$, and suppose that 
$e \in \vio{\tau}{\J}$. Since $\J$ is a model of $\Knew$, we have
$e \in A_\tau^\J$ (recall Lemma \ref{newkb}, item (v)). 
It follows from the definition of $\J$ that there exists 
$e^*$ such that $e=[e^*]_{\sim_{|q|+1}}$ and $e^* \in A_\tau^{\I^*}$. 
However, since $\I^*$ is also a model of $\Knew$, we can infer from 
$e^* \in A_\tau^{\I^*}$ that $e^* \in \vio{\tau}{\interlacingofstar{\I}}$. 
By Lemma \ref{interleaving-cost}, 
we know that $\vio{\tau}{\interlacingofstar{\I}} \subseteq \Delta^*$. 
As $e^* \in \Delta^*$, we have $[e^*]_{\sim_{|q|+1}} = \{e^*\}$. 
It follows that $\vio{\tau}{\J} \subseteq \{\{e^*\} \mid e^* \in \vio{\tau}{\interlacingofstar{\I}}\}$,
and hence $|\vio{\tau}{\J} | \leq |\vio{\tau}{\interlacingofstar{\I}}|$. 
Putting everything together, we can therefore conclude that 
$\cost{\wkb}{\J} \leq \cost{\wkb}{\interlacingofstar{\I}}$. 
\end{proof}

To complete the proof of Proposition \ref{prop:CQboundedCertain}, we now 
examine the steps in the construction in order to place a bound on 
the size of the interpretation $\J$. 
By analyzing the number of equivalence classes, the following 
upper bound on $|\Delta^\J|$ was shown in \cite{Maniere}:
\begin{align*}
(2|q|+3) \times |\T|^{|q|+2} \times (|\deltastar| + 2^{\signatureof[\T]} + 1)^{|\T|^{2|q|+2}}  
\end{align*}
It therefore remains to bound the set $|\Delta^*|$ (using our definition of $\deltastar$): 
\begin{align*}
|\Delta^*| &\leq  |\individuals(\Knew)| + |\bigcup_{\tau \in \T \setminus \Tinf} \vio{\tau}{\Idag}|\\
& \leq  |\individuals(\Knew)|+ \cost{\wkb}{\Idag} \\
& \leq |\individuals(\A)| + |\individuals(\T)| + \cost{\wkb}{\I} \\
& \leq |\individuals(\A)| + |\individuals(\T)| + k
\end{align*}
By substituting the latter quantity for $|\Delta^*|$ in the bound above, 
we obtain the desired upper bound on $|\Delta^\J|$ that is polynomial w.r.t. $k$ and $|\A|$, 
assuming $|\T|$ and $|q|$ are treated as constants.

\subsection{Lower Bounds}

\ThIQAoptDataLowerBound*
\begin{proof}
We use a reduction from the $\deltaptwolog$-complete problem of deciding whether a given vertex belongs to all the independent sets of maximum size of a graph \cite[Lemma 6]{DBLP:conf/icdt/LopatenkoB07,DBLP:journals/corr/LopatenkoB16}.
Let $\mathcal{G}=(\mathcal{V},\mathcal{E})$ be a graph and $w \in \mathcal{V}$ a vertex. 
Define an $\mathcal{EL}_\bot$ WKB $\wkb = (\langle \T, \A \rangle, \omega )$ as follows.
\begin{align*}
	\A &= \{In_i(v) \mid v \in \mathcal{V}, i\in\{1,2\} \} \\
			&\cup \{ Edge(x,y) \mid (x,y) \in \mathcal{E} \} \\
			&\cup \{ Distinguish(w) \}	\\
			&\cup \{ NoGoal(w) \} 	\\
	\T &= \{ In_i \sqcap \exists Edge.In_i \sqsubseteq \bot \mid i\in\{1,2\} \} \\
			&\cup  \{ In_i \sqcap Distinguish \sqsubseteq Goal \mid i\in\{1,2\} \} \\
			&\cup \{ Goal \sqcap NoGoal \sqsubseteq \bot \} \\
	\omega&(\tau) = \infty  \text{ for } \tau\in\T \\
	\omega&(In_i(v)) = 1  \text{ for }  v \in V, i\in\{1,2\}  \\
	\omega&(Edge(x,y)) = \infty \text{ for }  (x,y) \in \mathcal{E} \\
	\omega&(Distinguish(w)) = \infty \\				
	\omega&(NoGoal(w)) = 1
\end{align*}

We show that $w$ does not belong to all the independent sets of maximum size of $\mathcal{G}$ iff $\wkb \sat{p}{opt} NoGoal(w)$.

For every independent set $I \subseteq \mathcal{V}$ of $\mathcal{G}$, define the interpretation $\I_I=(\mathcal{V},\cdot^{\I_I})$ by
\begin{itemize}
\item $In_i^{\I_I} = I$, $i\in\{1,2\} $
\item $Edge^{\I_I} = \mathcal{E}$
\item $Distinguish^{\I_I} = \{ w \}$
\item $Goal^{\I_I} = \{w\}$ if $w\in I$ and $Goal^{\I_I} = \emptyset$ otherwise
\item $NoGoal^{\I_I} = \emptyset$ if $w\in I$ and $NoGoal^{\I_I} = \{ w\}$ otherwise
\end{itemize}
It is easy to check that $\cost{\wkb}{\I_I} = 2|\mathcal{V}\setminus I| + \delta_w^I$ where $\delta_w^I = 1$ if $w\in I$ and $\delta_w^I = 0$ otherwise. 

Moreover, for every interpretation $\I$ of finite cost, note that for $i\in\{1,2\} $, $In_i^\I \cap \mathcal{V}$ is an independent set of $\mathcal{G}$ (since $\I \models  In_i \sqcap \exists Edge.In_i \sqsubseteq \bot$).  \smallskip

\noindent($\Rightarrow$) Suppose there is  an independent set $M$ of $\mathcal{G}$ of maximal size that does not contain $w$. 
By construction, $\I_M \models NoGoal(w)$. We show that $\I_M$ has minimal cost. 
Suppose for a contradiction that there exists $\I$ such that $\cost{\wkb}{\I} < \cost{\wkb}{\I_M}$. Then $|\mathcal{V} \setminus In_1^\I|+|\mathcal{V} \setminus In_2^\I| \leq \cost{\wkb}{\I} < 2|\mathcal{V} \setminus M|$ (since $\delta_w^M = 0$).
We then have that $|In_1^\I\cap\mathcal{V}| > |M|$ or $|In_2^\I\cap\mathcal{V}| > |M|$ which contradicts the fact that $M$ is of maximum size. 
Therefore $\I_M$ has optimal cost and $\wkb \sat{p}{opt} NoGoal(w)$. \smallskip

\noindent($\Leftarrow$) Suppose that $\wkb \sat{p}{opt} NoGoal(w)$, \ie there is $\I$ with optimal cost such that $\I \models NoGoal(w)$. 
Since $\I$ has finite cost, $In_1^\I\cap\mathcal{V}$ and $In_2^\I\cap\mathcal{V}$ are  independent sets of $\mathcal{G}$. 
Suppose for a contradiction that there is an independent set $M$ of $\mathcal{G}$ such that $|M|>\max(|In_1^\I\cap\mathcal{V}|,|In_2^\I\cap\mathcal{V}|)$. Then $\cost{\wkb}{\I_M} \leq 2|\mathcal{V} \setminus M| + 1 < |\mathcal{V} \setminus In_1^\I|+|\mathcal{V} \setminus In_2^\I|$ i.e. $\cost{\wkb}{\I_M} < \cost{\wkb}{\I}$ (since $\I \models NoGoal(w)$), which contradicts our assumption. 
Hence $In_1^\I\cap\mathcal{V}$ or $In_2^\I\cap\mathcal{V}$ is an independent set of $\mathcal{G}$ of maximum size that does not contain~$w$.

We may therefore conclude that \dec{I}{opt}{p} for $\mathcal{EL}_\bot$ is $\deltaptwolog$-hard in data complexity (even with weights bounded independently from $|\A|$ and encoded in unary). 

\medskip

We now consider \dec{I}{opt}{c} and show that $w$ belongs to all the independent sets of maximum size of $\mathcal{G}$ iff $\wkb \sat{c}{opt} Goal(w)$.  \smallskip

\noindent($\Rightarrow$) 
Suppose that $w$ belongs to all independent sets of $\mathcal{G}$ of maximum size. 
Let $\I$ be an interpretation with optimal cost \wrt $\wkb$. 
Since $\I$ has finite cost, then for $i\in\{1,2\}$, $In_i^\I\cap\mathcal{V}$ is an independent set of $\mathcal{G}$. 
Assume for a contradiction that neither $In_1^\I\cap\mathcal{V}$ nor $In_2^\I\cap\mathcal{V}$ is of maximal size and let $M \subseteq \mathcal{V}$ be an independent set of $\mathcal{G}$ of maximum size. Then $w\in M$ and $\cost{\wkb}{\I_M} = 2|\mathcal{V} \setminus M| + 1 < \cost{\wkb}{\I}$, contradicting the optimal cost of $\I$. Therefore there is $i\in\{1,2\}$ such that $In_i^\I\cap\mathcal{V}$ is an independent set of maximum size, which implies that $w \in In_i^\I$. Finally, as $\I \models In_i \sqcap Distinguish \sqsubseteq Goal$ and $\I \models  Distinguish(w)$,  we have that $\I \models Goal(w)$. We have thus shown that  $\wkb \sat{c}{opt} Goal(w)$. \smallskip

\noindent($\Leftarrow$) Let $\wkb \sat{c}{opt} Goal(w)$, and 
suppose for a contradiction that there exists an independent set $M \subseteq \mathcal{V}$ of $\mathcal{G}$ of maximal size such that $w \not\in M$. We show that $\I_M$ has optimal cost. Indeed, suppose $\optc{\wkb} < \cost{\K}{\I_M}$ and let $\I$ be an interpretation of optimal cost. As $\wkb \sat{c}{opt} Goal(w)$ we have that $w \in Goal^\I$. Thus there must exist $i\in\{1,2\}$ such that $w \in In_i^\I$ (otherwise the interpretation $\I'$ obtained from $\I$ by removing $w$ from $Goal^\I$ and adding $w$ to $NoGoal^{\I}$ would have a smaller cost than $\I$). Therefore $\cost{\wkb}{\I} = |\mathcal{V} \setminus In_1^\I|+|\mathcal{V} \setminus In_2^\I| +1 < \cost{\wkb}{\I_M} = 2|\mathcal{V} \setminus M| $ since $\delta_w^M = 0$ as $w\notin M$.
Thus $\min(|\mathcal{V} \setminus In_1^\I|,|\mathcal{V} \setminus In_2^\I| )< |\mathcal{V} \setminus M|$ and we deduce $|In_i^\I\cap\mathcal{V}| > |M|$ for some $i\in\{1,2\}$ which is a contradiction since $In_i^\I\cap\mathcal{V}$ is an independent set. 
Hence $\I_M$ is an interpretation of optimal cost w.r.t $\wkb$ and $w\in Goal^{\I_M}$. 
It follows that $w \in In_i^{\I_M}$ for $i\in\{1,2\}$ (with the same argument as before), \ie $w\in M$. 
Thus $w$ belongs to all independent sets of $\mathcal{G}$ of maximum size. 

We have thus shown that \dec{I}{opt}{c} for $\mathcal{EL}_\bot$ is $\deltaptwolog$-hard in data complexity (even with weights are bounded independently from $|\A|$ and encoded in unary).
\end{proof}

\ThCQAoptDataLowerBound*
\begin{proof}
The proof is by reduction from the following $\Delta^P_2$-hard problem: given a satisfiable boolean formula $\varphi$ over variables $x_1,\dots, x_n$ such that $\varphi$ is a set of clauses $c_1,\dots, c_m$ where each clause has exactly two positive and two negative literals  and any of the four positions in a clause can be filled instead by one of the truth constants true and false, and given  $k\in\{1,\dots, n\}$, decide whether the lexicographically maximum truth assignment $\nu_{max}$ satisfying $\varphi$ with respect to $(x_1,\dots, x_n)$ fulfills $\nu_{max}(x_k)=1$. It follows from \cite{Krentel88} and from the reductions from SAT to 3SAT and from 3SAT to 2+2SAT \cite{Schaerf93} that this problem is $\Delta^P_2$-hard.

Let $\varphi$ be an instance of the problem as previously defined. We define $\WKB$ an \ELbot WKB as follows:

\begin{align*}
\A=&\{S(\varphi,c_j)\mid 1\leq j\leq m\}\cup\\
&\{P_\ell(c_k,x_j) \mid \ell\in\{1, 2\},x_j\text{ is the }\ell^{th}\text{ pos. lit. of }c_k \}\cup\\
    &\{N_\ell(c_k,x_j) \mid \ell\in\{1, 2\},\neg x_j\text{ is the }\ell^{th}\text{ neg. lit. of }c_k \}\cup\\
    & \{F(x_i), T(x_i), T'(x_i)\mid 1\leq i\leq n\}\cup\{F(f), T(t)\}
    \\
\T=&\{F\sqcap T\sqsubseteq \bot, F\sqcap T'\sqsubseteq \bot\}\cup\\
&\{  \exists S.(  \exists P_1.F\sqcap\exists P_2.F\sqcap \exists N_1.T\sqcap \exists N_2.T) \sqsubseteq \bot \}
\end{align*}

We set $\omega(\tau)=\infty$ for every $\tau\in\T$, and $\omega(\alpha)=\infty$ for every  assertion $\alpha$ built on $P_\ell$ and $N_\ell$ as well as for $T(t)$ and $F(f)$. For the remaining assertions, we define the weights through the following prioritization (see Lemma~6.2.5 in \cite{DBLP:phd/hal/Bourgaux16}):
\begin{itemize}
\item $L_1=\{T(x_i), F(x_i)\mid 1\leq i\leq n\}$
\item $L_2=\{S(\varphi,c_j)\mid 1\leq j\leq m\}$
\item $L_3=\{T'(x_1)\}$, $L_4=\{T'(x_2)\}$, $\dots$, $L_{n+2}=\{T'(x_n)\}$
\end{itemize} 
Let $u=max(2n,m)+1$ and define $\omega(\alpha)=u^{n+2-i}$ for every $\alpha\in L_i$. 

We show that $\wkb \sat{c}{opt} T'(x_k)$ iff $\wkb \sat{p}{opt} T'(x_k)$ iff $\nu_{max}(x_k)=1$.

Let $\I$ be an interpretation with optimal cost. By definition of $\omega$, $\I\models\T$ and $\I$ satisfies all $P_\ell$/$N_\ell$ assertions, $T(t)$ and $F(f)$. 
First note that for every $1\leq i\leq n$, $\I$ satisfies exactly one of the assertions $T(x_i), F(x_i)$. Indeed, it cannot contain both since $\I\models F\sqcap T\sqsubseteq \bot$ and if there was $i$ such that $\I\not\models T(x_i)$ and $\I\not\models F(x_i)$, then $\J$ obtained from $\I$ by setting $T^\J=T^\I\cup\{x_i^\I\}$, $S^\J=\emptyset$ and $X^\J=X^\I$ for every other concept or role name $X$ would be such that $\cost{\wkb}{\J}<\cost{\wkb}{\I}$ since $m\times u^{n} < u^{n+1}$ because $m<u$. 
Moreover, it is easy to see that $\I\models T'(x_i)$ exactly when $\I\not\models F(x_i)$, i.e.\ when $\I\models T(x_i)$. 

Let $\nu$ be the valuation of $x_1,\dots, x_n$ such that $\nu(x_i)=1$ iff $\I\models T(x_i)$ (iff $\I\models T'(x_i)$ iff $\I\not\models F(x_i)$). 
Assume for a contradiction that $\nu$ does not satisfy $\varphi$. There exists a clause $c_j$ of $\varphi$ such that for every positive literal $x_i$ of $c_j$, $\nu(x_i)=0$, i.e.\ $\I\models F(x_i)$, and for every negative literal $\neg x_i$ of $c_j$, $\nu(x_i)=1$, i.e.\ $\I\models T(x_i)$. It follows that $\I\not\models S(\varphi, c_j)$ (otherwise $\varphi^\I$ would be in $( \exists S.(  \exists P_1.F\sqcap\exists P_2.F\sqcap \exists N_1.T\sqcap \exists N_2.T))^\I$ and the cost of $\I$ would be infinite. 
Since $\varphi$ is satisfiable, there exists $\nu'$ that satisfies $\varphi$. Let $\J$ be the interpretation such that $P_\ell^\J=P_\ell^\I$ and $N_\ell^\J=N_\ell^\I$ for $1\leq \ell\leq 2$, $T^\J=\{t^\I\}\cup\{x_i^\I\mid \nu'(x_i)=1\}$, $F^\J=\{f^\I\}\cup\{x_i^\I\mid \nu'(x_i)=0\}$, and  $S^\J=\{(\varphi,c_j)\mid 1\leq j\leq m\}$. Then $\cost{\wkb}{\J}<\cost{\wkb}{\I}$. 
Indeed, $\I$ and $\J$ both violate the same number of $T$/$F$-assertions and the cost $u^n$ of violating the single assertion $S(\varphi, c_j)$ is greater than the cost of violating all $T'$-assertions $\Sigma_{i=3}^{n+2} u^{n+2-i}=\Sigma_{i=0}^{n-1} u^{i}=\frac{u^n-1}{u-1}$. 
Hence, $\nu$ satisfies $\varphi$.

Finally, we show that $\nu=\nu_{max}$. Define $\J$ by $P_\ell^\J=P_\ell^\I$ and $N_\ell^\J=N_\ell^\I$ for $1\leq \ell\leq 2$, $T^\J={T'}^\J=\{t^\I\}\cup\{x_i^\I\mid \nu_{max}(x_i)=1\}$, $F^\J=\{f^\I\}\cup\{x_i^\I\mid \nu_{max}(x_i)=0\}$, $S^\J=\{(\varphi,c_j)\mid 1\leq j\leq m\}$. If there exists $l$ such that $\nu$ and $\nu_{max}$ coincide on $x_i$ for $i<l$ but not on $x_l$, i.e.\ $\nu(x_l)=0$ and $\nu_{max}(x_l)=1$ (since $\nu_{max}$ is the lexicographically maximum truth assignment satisfying $\varphi$), then $\cost{\wkb}{\J}<\cost{\wkb}{\I}$. Indeed, the cost $u^{n-l}$ of violating $T'(x_l)$ is greater than the cost $\Sigma_{i=l+3}^{n+2} u^{n+2-i}=\Sigma_{i=0}^{n-l-1} u^{i}=\frac{u^{n-l}-1}{u-1}$ of violating all $T'(x_i)$ for $i>l$. 
Hence $\nu$ and $\nu_{max}$ coincide on every $x_i$, $1\leq i\leq n$.

We have shown that for every interpretation with optimal cost $\I$, $\I\models T'(x_k)$ iff $\nu_{max}(x_k)=1$. Hence $\wkb \sat{c}{opt} T'(x_k)$ iff $\wkb \sat{p}{opt} T'(x_k)$ iff $\nu_{max}(x_k)=1$.
\end{proof}

\end{document}